\documentclass[12pt, leqno]{amsart}

\usepackage{graphicx}
\usepackage{amssymb,amsmath,amsthm,amscd}
\usepackage{mathrsfs}
\usepackage{lscape}
\usepackage{enumerate}
\usepackage[usenames,dvipsnames]{color}
\usepackage[colorlinks=true, pdfstartview=FitV, linkcolor=blue,citecolor=blue,urlcolor=blue]{hyperref}

\usepackage{tikz}
\usepackage[all]{xy}
\usetikzlibrary{matrix}

\usepackage{mathtools}

\DeclarePairedDelimiter\floor{\lfloor}{\rfloor}

\allowdisplaybreaks[3]

\usepackage{comment}
\usepackage{amsmath}
\newcommand{\Mod}[1]{\ (\mathrm{mod}\ #1)}

\usepackage{color}

\numberwithin{equation}{section}

\newcommand{\ber}{\begin{red}}
\newcommand{\er}{\end{red}}
\newcommand{\beb}{\begin{blue}}
\newcommand{\eb}{\end{blue}}

\theoremstyle{plain}
\newtheorem{lemma}{Lemma}[section]
\newtheorem{proposition}[lemma]{Proposition}
\newtheorem{theorem}[lemma]{Theorem}

\newtheorem{corollary}[lemma]{Corollary}

\theoremstyle{definition}
\newtheorem{remark}[lemma]{Remark}
\newtheorem{example}[lemma]{Example}

\newtheorem{definition}[lemma]{Definition}
\newtheorem{tdefinition}[lemma]{Theorem-Definition}

\newcommand{\rvline}{\hspace*{-\arraycolsep}\vline\hspace*{-\arraycolsep}}

\newcommand{\n}{\mathfrak{n}}
\newcommand{\g}{\mathfrak{g}}

\newcommand{\Z}{\mathbb{Z}}

\newlength{\mylength}
\def\fg{\mathfrak{g}}
\def\fh{\mathfrak{h}}
\def\calW{\mathcal{W}}

\title{Generators of supersymmetric classical $W$-algebras}

\author
[E. Ragoucy]{Eric Ragoucy}
\address[E. Ragoucy]{Laboratoire de Physique Th\'{e}orique LAPTh,
CNRS, Universit\'{e} Savoie Mont Blanc and U.G.A.,
BP 110, 74941 Annecy-le-Vieux Cedex, France}
\email{eric.ragoucy@lapth.cnrs.fr}
\author
[A. Song]{Arim Song}
\address[A. Song]{Department of Mathematical Sciences, Seoul National University, GwanAkRo 1, Gwanak-Gu, Seoul 08826}
\email{ireansong@snu.ac.kr}
\author[U.R. Suh]{Uhi Rinn Suh}
\address[U.R. Suh]
{ Department of Mathematical Sciences and Research institute of Mathematics, Seoul National University, GwanAkRo 1, Gwanak-Gu, Seoul 08826,
Korea}
\email{uhrisu1@snu.ac.kr}

\date{\today}

\begin{document}

\begin{abstract}

Let $\mathfrak{g}$ be a Lie superalgebra of type $\mathfrak{sl}$ or $\mathfrak{osp}$ with an odd principal nilpotent element $f$. We consider a matrix $\mathcal{A}_{\mathfrak{g},f}$ determined by $\mathfrak{g}$ and $f$ and find a generating set of the supersymmetric classical $W$-algebra $\mathcal{W}(\bar{\mathfrak{g}},f)$ using the row determinant of $\mathcal{A}_{\mathfrak{g},f}$. 

\end{abstract}

\maketitle


\section{Introduction}

The notion of $W$-algebra appeared first in physics, in the works of Zamolodchikov \cite{Zam} and of Lukyanov and Fateev \cite{FatLuky}. In this context, it appeared as an extension of the Virasoro algebra, which is a symmetry of two-dimensionnal conformal field theories. It was then realized that the classical version of $W$-algebras was the new Poisson algebra introduced in \cite{DrSo} by Drinfeld and Sokolov in the study of the center of an affine Lie algebra. Generalizations of classical $W$-algebras  were introduced in relation to 
Hamiltonian reduction of Wess-Zumino-Witten (WZW) to Toda models \cite{BFOFW,BTvD91}. 
They were noted $W(\fg,\fh)$ where $\fg$ is the finite dimensional Lie algebra of the underlying affine Lie algebra. 
 $\fh$ is a subalgebra of $\fg$, defining a nilpotent element in $\fg$. The nilpotent element is defined as the principal nilpotent of $\fh$, so that the $W$-algebras defined by Drinfeld and Sokolov were just the classical version of $W(\fg,\fg)\equiv W(\fg)$.

The case where $\fg$ is a superalgebra was studied  first in \cite{EvHo,KoMoNo}, and led to $W$-superalgebras. Then, in \cite{DeRaSo}, the superfield formalism was introduced to deal with the supersymmetric (SUSY) version of $W$-algebras. A complete classification of classical $W$-algebras and of SUSY classical $W$-algebras  can be found in \cite{FrRaSo}. For classical $W$-algebras and superalgebras, it relies on the classification of $\mathfrak{su}(2)$ embeddings  in $\fg$ \cite{Dynkin}. For SUSY $W$-algebras, it uses the classification of $\mathfrak{osp}(1|2)$ embeddings \cite{LeiSaSe,FrRaSo}  in $\fg$.
Reviews on the subject can be found in \cite{FORTW,BoSch}.

The quantum version of these algebras is done in the framework of BRST cohomology, which deals, at the quantum level, of the Hamiltonian reduction introduced in the classical case. They appeared first in \cite{FeFr}, and the cohomology was resolved in 
 \cite{deBTj} for quantum $W$-algebras. The supersymmetric version was done in \cite{MR94}.

BRST cohomologies can be understood in the theory of vertex algebras 
and  vertex algebra structures of quantum $W$-algebras are written in  \cite{KRW, KW1, DK06}. On the other hand, classical $W$-algebras can be unerstood as quasi-classical limit of quantum $W$-algebras so that they are Poisson vertex algebras \cite{deSoKaVa, DKV06, Suh16}. In these articles, the authors often denote by $W(\g, F,k)$ and $\mathcal{W}(\g,F,k)$ the quantum and classical $W$-algebras $W(\g, \fh)$  and $\mathcal{W}(\g, \fh)$ where  $F$ is the \textit{even} principal nilpotent element of $\fh$ and  $k$ is the central element of the affine Lie superalgebra $\hat{\g}$. For the principal nilpotent element $F$ of $\g$, an explicit construction of  generators of $W(\g,F,k)$ was built in \cite{AraMo} when $\fg=\mathfrak{sl}(N)$. Analogous results on $\mathcal{W}(\g, F,k)$ were written in \cite{MR} when $\fg=\mathfrak{sl}(N), \mathfrak{so}(N),
\mathfrak{sp}(2n)$ and $G_2$.
See also \cite{deSole14} for a review on the subject.

The SUSY quantum $W$-algebras \cite{MRS20} and the SUSY classical $W$-al\-ge\-bras \cite{Suh} are defined as SUSY vertex algebras and SUSY Poisson vertex algebras \cite{Kac, HK06}. Since  a SUSY  $W$-algebra is governed by a Lie superalgebra $\g$ and an $\mathfrak{osp}(1|2)$ embedding which is determined by an \textit{odd} nilpotent element $f$, quantum and classical SUSY $W$-algebras  can be denoted 
 by $W(\bar{\g}, f,k)$ and $\mathcal{W}(\bar{\g}, f,k)$. In addition, in \cite{MRS20},
 explicit forms of generators of the SUSY quantum $W$-algebra $W(\overline{\mathfrak{sl}}(n|n\pm1),f,k)$ has been constructed when $f$ is odd principal nilpotent.  

In this paper, we consider a  finite simple Lie superalgebra of type $\mathfrak{sl}$ or $\mathfrak{osp}$ which possesses an odd principal nilpotent element $f$. Such a Lie superalgebra $\g$ is of type $\mathfrak{sl}(n|n\pm1)$ or of type $\mathfrak{osp}(M|2n)$ with $M=2n\pm1,\,2n+2,
\, 2n$. We provide an explicit constructions of generators of SUSY classical $W$-algebras $\mathcal{W}(\bar{\g}, f)=\mathcal{W}(\bar{\g}, f,k)|_{k=1} $ via row-determinants of some matrices. In the case of $\g=\mathfrak{osp}(M|2n)$ with $M=2n\pm 1$, the construction relies on a folding procedure, used in the context in \cite{folding} to determine the structure of $W$-algebras based on $\fg=\mathfrak{so}(N), \mathfrak{sp}(2n)$ or $\mathfrak{osp}(M|2n)$ Lie (super)algebras. We rephrase it in the framework of SUSY Poisson vertex algebras and find generators of $\mathcal{W}(\bar{\g}, f)$.

The paper is organized as follows. Section \ref{sec:prel} reminds basic results on  Lie superalgebras and SUSY Poisson vertex algebras, that are needed in the following sections. Then, in Section \ref{sec:SUSYW},
we recall the general structure theory of SUSY classical $W$-algebras. Our results on the  generators for  SUSY classical  $W$-algebras are presented in Section \ref{sec:glmn}, \ref{sec:ospodd} and  \ref{sec:ospeven}. 
In each section, we deal with 
$\calW(\overline{\mathfrak{gl}}(n|n\pm1),f)$, $\calW(\overline{\mathfrak{osp}}(2n\pm1|2n),f)$  and  $\calW(\overline{\mathfrak{osp}}(M|2n),f)$ with $M=2n,2n+2$.

\section{Preliminaries\label{sec:prel}}
In this section, we review definitions and properties of Lie superalgebras and supersymmetric  Poisson vertex algebras. For more details, we refer to  
Chapter 1 of \cite{Cheng12} and Section 4 of \cite{HK06}.

 Throughout this paper, we assume that the base field is $\mathbb{C}$.

\subsection{Lie superalgebras} \ 

A vector superspace $V=V_{\bar{0}} \oplus V_{\bar{1}}$ is a vector space with  a $\mathbb{Z}/2\mathbb{Z}$-grading. 
An element $a \in V_{i}$ is called even (resp. odd) if $i = \bar0$ (resp. $i=\bar1$). The parity $p(a)$ of an even (resp. odd) element $a$ is defined by $p(a)=0$ (resp. $p(a)=1$).
An algebra $A=A_{\bar{0}} \oplus A_{\bar{1}}$ is a superalgebra if  it is a vector superspace satisfying 
$A_{i} A_{j} = A_{i+j}$ for $i,j \in \{\bar0,\bar1\}$.

\begin{definition}
Let $\g$ be a vector superspace endowed with a $\mathbb{Z}/2\mathbb{Z}$-grade preserving bilinear bracket 
\[\  [ \ , \ ] : \g \times \g \to \g.\]
If the bracket satisfies 
\begin{itemize}
\item (skewsymmetry) $[a,b]=- (-1)^{p(a)p(b)}[b,a]$,
\item (Jacobi identity) $[a,[b,c]]= [[a,b],c]+ (-1)^{p(a)p(b)}[b,[a,c]]$,
\end{itemize}
for any $a,b,c\in\g$, then $\g$ is called a {\it Lie superalgebra.}
\end{definition}

For a vector superspace $V$, the algebra $\text{End}(V)$ of endomorphisms on $V$ is a superalgebra such that 
$\text{End}(V)_{i} = \{f \in \text{End}(V)| f(V_{j})\subset V_{i+j}\}$ for $i,j=\bar0,\bar1$.
One can define a Lie superalgebra structure on $\text{End}(V)$ by considering the super commutator 
\begin{equation}\label{commutator}
 [f,g]=fg -(-1)^{p(f)p(g)} gf \quad  \text{ for } f,g \in \text{End}(V).
 \end{equation}
Such Lie superalgebra is called the {\it general linear Lie superalgebra} of $V$ and denoted by $\mathfrak{gl}(V)$. 
When the dimensions of $V_{\bar{0}}$ and $V_{\bar{1}}$ are $m$ and $n$, respectively,  $\mathfrak{gl}(V)$ is also denoted by $\mathfrak{gl}(m|n)$.

Let us describe $\mathfrak{gl}(m|n)$ more explicitly. Let $I=\{1,2,\cdots, m+n\}$  and let  $V$ be a vector superspace 
with a basis $\{v_i|i\in I\}$ which consists of $m$ even elements and $n$ odd elements. Then there are subsets $I_{\bar{0}}$ and $I_{\bar{1}}$ of $I$ such that  $\{v_i|i\in I_{\bar{0}}\}$ and $\{v_i|i\in I_{\bar{1}}\}$ are bases of $V_{\bar{0}}$ and $V_{\bar{1}}$. 
Then 
\begin{equation}
I= I_{\bar{0}} \sqcup I_{\bar{1}}.
\end{equation}
Note that $\mathfrak{gl}(m|n)$ is spanned by 
$ \bold{e}_{ij}$ such that $\bold{e}_{ij} v_k= \delta_{jk} v_i$ for $i,j,k \in I$.
For convenience, denote
\begin{equation} \label{parity_index}
p(i):= p(v_i).
\end{equation}
Then the parity $p(\bold{e}_{ij})\in\{0,1\}$ is defined by  
$(-1)^{p(\bold{e}_{ij})}= (-1)^{p(i)+p(j)}.$
The bracket \eqref{commutator} and the even supersymmetric invariant bilinear form $(\ |\ )$ on $\mathfrak{gl}(m|n)$ are written as follows:
\begin{equation} \label{Eqn:gl(m|n)}
   [ \bold{e}_{ij} , \bold{e}_{kl}]= \delta_{j,k} \bold{e}_{il} -(-1)^{(p(i)+p(j))(p(k)+p(l))} \delta_{i,l} \bold{e}_{kj}, \quad (\bold{e}_{ij}|\bold{e}_{kl})= \delta_{jk} \delta_{il}(-1)^{p(i)}.
\end{equation}
In addition, the {\it supertrace} of $a= \sum_{i\in I} a_{ij} \bold{e}_{ij}\in \mathfrak{gl}(m|n)$ for $a_{ij} \in \mathbb{C}$ is 
\begin{equation}
str(a)= \sum_{i\in I} (-1)^{p(i)} a_{ii}\in \mathbb{C},
\end{equation}
and the {\it supertranspose} of $a$ is 
\begin{equation}\label{supertrace}
 a^{st} = \sum_{i,j\in I} (-1)^{(p(i)+1)p(j)} a_{ij} \bold{e}_{ji}\in\mathfrak{gl}(m|n). 
\end{equation}

If we assume $p(i)=0$ for $i=1,2, \cdots, m$ and $p(i)=1$ for $i=m+1, m+2, \cdots, m+n$, then 
\begin{equation}  \label{gl(m|n)_standard}
\mathfrak{gl}(m|n)=  
\left\{ \left.
\begin{pmatrix}
  \begin{matrix}
A
  \end{matrix}
  & \rvline & B \\
\hline
C & \rvline &
  \begin{matrix}
D
  \end{matrix}
\end{pmatrix}\right|\left.\begin{array}{cc} A: m \times m,  &  B:m\times n, \\  C: n\times m, &   D: n\times n \end{array} \right.
 \text{ matrices }
\right\},
\end{equation}
where $A$ and $D$ correspond to the even part and $B$ and $C$ to the odd part.

\begin{example} \label{Ex: sl and osp} 
The following examples are simple Lie super-subalgebras of general linear Lie superalgebras.
\begin{enumerate}
\item Let $n\neq m$ and $m,n\geq 1$. The special linear Lie superalgebra of type $A(m|n)$ is 
\begin{equation} \label{sl}
\mathfrak{sl}(m+1|n+1)= \{ X\in \mathfrak{gl}(m+1|n+1)| str(X)=0 \}.
\end{equation}
\item Consider $I=\{1,2, \cdots, 2m+2n+1\}$ for $m\geq 0$ and $n\geq 1$ and assume that 
$p(i)=0$ for $i=1,2, \cdots, 2m+1$ and $p(i)=1$ for $i=2m+2,2m+3, \cdots, 2m+2n+1$. 
Then the ortho-symplectic Lie superalgebra $\mathfrak{osp}(2m+1|2n)\subset \mathfrak{gl}(2m+1|2n)$ is defined by 
\begin{equation}\label{osp(2m+1|2n)}
 \mathfrak{osp}(2m+1|2n)=\{ \, a \in \mathfrak{gl}(2m+1|2n) \,  | \, a^{st} \mathcal{J}^{stan}_{2m+1|2n} + \mathcal{J}^{stan}_{2m+1|2n} a=0 \, \}, 
 \end{equation}
 where 
\begin{equation}
\mathcal{J}^{stan}_{2m+1|2n} = 
\begin{pmatrix}
  \begin{matrix}
0 & I_m & 0 \\
I_m & 0 & 0 \\
0 & 0 & 1
  \end{matrix}
  & \rvline & 0  \\
\hline
0 & \rvline &
  \begin{matrix}
0 & I_n \\
-I_n & 0
  \end{matrix}
\end{pmatrix}\in \mathfrak{gl}(2m+1|2n).
\end{equation}
 Here, $I_k$ stands for the identity matrix of size $k\times k$. The Lie superalgebra \eqref{osp(2m+1|2n)} is called type $B(m,n)$.

\item  Consider $I=\{1,2,\cdots, 2m+2n\}$ for $m\geq 1$ and $n\geq 1$ and assume that 
$p(i)=0$ for $i=1,2, \cdots, 2m$ and $p(i)=1$ for $i=2m+1,2m+2, \cdots, 2m+2n$.  
Then the ortho-symplectic Lie superalgebra $\mathfrak{osp}(2m|2n)\subset \mathfrak{gl}(2m|2n)$ is defined by 
\begin{equation}\label{osp(2m|2n)}
 \mathfrak{osp}(2m|2n)=\{ \, a \in \mathfrak{gl}(2m|2n) \,  | \, a^{st} \mathcal{J}^{stan}_{2m|2n} + \mathcal{J}^{stan}_{2m|2n} a=0 \, \},\end{equation}
where
 \begin{equation}
\mathcal{J}^{stan}_{2m|2n} = 
\begin{pmatrix}
  \begin{matrix}
0 & I_m \\
-I_m & 0 \\
  \end{matrix}
  & \rvline & 0  \\
\hline
0 & \rvline &
  \begin{matrix}
0 & I_n \\
-I_n & 0
  \end{matrix}
\end{pmatrix}.
\end{equation}
 If $m=1$, it is called  type $C(n+1)$ and if $m\geq 2$, it is called type $D(m,n)$.
\end{enumerate}
\end{example}

In the Lie superalgebras in Example \ref{Ex: sl and osp}, the bilinear form 
\begin{equation}
(X|Y):= str (XY)
\end{equation}
is nondegenerate, even, supersymmetric and invariant.

\subsection{Lie superalgebras with odd principal nilpotent elements}

\begin{proposition} \cite{LeiSaSe,FrRaSo}\label{Prop:OSRS} 
Let $\g$ be a finite simple Lie superalgebra isomorphic to $\mathfrak{sl}(m|n)$ or $\mathfrak{osp}(m|2n)$ for $m,n \geq 1$. If $\g$ has an odd principal nilpotent element then $\g$ is isomorphic to one of the following algebras: 
\[ \mathfrak{sl}(n\pm 1|n), \quad \mathfrak{osp}(M| 2n) \text{ with } M=2n\pm1,\ 2n+2, \ 2n, \quad  
 D(2,1;\alpha) \text{ with } \alpha \in \mathbb{C}\setminus \{0,\pm 1\}.\]
\end{proposition}

For the rest of the paper, we deal with Lie superalgebras of type $\mathfrak{gl}$,  $\mathfrak{sl}$ and $\mathfrak{osp}$  possessing odd principal nilpotent elements.

Now, let us describe matrix presentations of Lie superalgebras in Proposition~\ref{Prop:OSRS}.
They possess a odd principal nilpotent element that is expressed as a triangular matrix.

\begin{example}\label{gl(n+1)}
Let $I=\{1,2, \cdots, 2n+1\}=I_{\bar{0}}\sqcup I_{\bar{1}}$ where \[ I_{0}= \{1,3,\cdots, 2n+1\} \text{ and } I_{\bar{1}}= \{2,4,\cdots, 2n\}.\] 
The Lie superalgebras $\mathfrak{gl}(n+1|n)$ and $\mathfrak{sl}(n+1|n)$ are defined by  \eqref{Eqn:gl(m|n)}
and
\eqref{sl}.Then 
\begin{equation}
 f= \sum_{i=1}^{2n} \bold{e}_{i+1\, i}
 \end{equation}
is an odd principal nilpotent element.
\end{example}


\begin{example} \label{gl(n-1)}
Let $I=\{1,2, \cdots, 2n-1\}=I_{\bar{0}}\sqcup I_{\bar{1}}$ where \[ I_{0}= \{2,4,\cdots, 2n-2\} \text{ and } I_{\bar{1}}= \{1,3,\cdots, 2n-1\}.\] 
The Lie superalgebras $\mathfrak{gl}(n-1|n)$ and $\mathfrak{sl}(n-1|n)$ are defined by  \eqref{Eqn:gl(m|n)} and \eqref{sl}. Then 
\begin{equation}
 f= \sum_{i=1}^{2n-2} \bold{e}_{i+1\, i}
 \end{equation}
is an odd principal nilpotent element.
\end{example}

\begin{example} \label{osp_4n pm1}
 Let $I= \{1,2, \cdots, M+2n\}$ for $M=2n-1 \text{ or }2n+1$. If $M=2n+1$ (resp. $M=2n-1$), we define $I_{\bar{0}}$ and $I_{\bar{1}}$ as in Example \ref{gl(n+1)} (resp. Example \ref{gl(n-1)}). 
 We introduce the map
\begin{equation}
\delta:I \to \{ 0,1\} , \quad k\mapsto \delta_k
\end{equation}
such that 
\begin{itemize}
\item if $M=2n+1$, \quad $\delta_k= \left\{\begin{array}{ll} 1 & \text{ if } k\in \{2n+2, 2n+4, \cdots, 4n\}, \\ 0 & \text{ otherwise;} \end{array} \right.$
\item if $M=2n-1$, \quad $\delta_k= \left\{\begin{array}{ll} 1 & \text{ if } k\in  \{2n+1, 2n+3, \cdots, 4n-1\}, \\ 0 & \text{ otherwise}. \end{array} \right.$
\end{itemize}
Then  $ \mathcal{J}_{M|2n}$ in Example \ref{Ex: sl and osp}  can be written as
\begin{equation} 
 \mathcal{J}_{M|2n}= \sum_{k\in I} (-1)^{\delta_k} \bold{e}_{k\, k'} \text{ for } k'=|I|+1-k.
\end{equation}
Consider the automorphism of $\mathfrak{gl}(M|2n)$ defined by
\begin{equation} \label{Theta}
\theta: \bold{e}_{ij} \mapsto  \tau(i,j) \bold{e}_{j'\, i'}:=(-1)^{p(i)p(j)+p(j)+1+\delta_i + \delta_j } \bold{e}_{j'\, i'}.
\end{equation}
Then, the Lie superalgebra $\mathfrak{osp}(M|2n)$ is spanned by 
\begin{equation} \label{F}
F_{ij}=\bold{e}_{ij} +\theta(\bold{e}_{ij}), \quad i,j\in I.
\end{equation}
For $i,j\in I$, $F_{j' \, i'}= \tau(i,j) F_{ij}$
and $ \{F_{ij}\}_{(i,j)\in \mathcal{B}}$ is a basis of $\mathfrak{osp}(M|2n)$, where
\begin{equation}\label{basis of osp}
\mathcal{B}= \{ (i,j) |\  i+j \leq |I|+1 \ \} \setminus \{ (i,i') |\ p(i)=0 \}\subset I \times I.
\end{equation}
Finally, 
\begin{equation}
f= \sum_{i=1}^{\frac{|I|-1}{2} } F_{i+1\, i}
\end{equation}
is an odd  principal nilpotent element.
\end{example}

\begin{example} \label{osp_4n +0,2}
 Let $I= \{1,2, \cdots, M+2n\}$ for $M=2n \text{ or }2n+2$.  Then $I_{\bar{0}}$ and $I_{\bar{1}}= I\setminus I_{\bar{0}}$ are defined by 
\begin{itemize}
\item if $M=2n, \qquad  I_{\bar{0}}=\{2k\,|\,1\leq k \leq n\} \cup \{ 2n+2k-1\,|\, 1 \leq k \leq n\},   $\\
\item  if $M=2n+2, \quad  I_{\bar{0}}=\{2k-1\,|\,1\leq k \leq n+1\} \cup \{ 2n+2k\,|\, 1 \leq k \leq n+1\}.$
\end{itemize}
We define the map
\begin{equation}
\delta:I \to \{ 0,1\} , \quad k\mapsto \delta_k
\end{equation}
such that 
\begin{itemize}
\item if $M=2n$, \qquad $\delta_k= \left\{\begin{array}{ll} 1 & \text{ if } k\in \{1,3, \cdots, 2n-1\}, \\ 0 & \text{ otherwise;} \end{array} \right.$
\item if $M=2n+2$, \quad $\delta_k= \left\{\begin{array}{ll} 1 & \text{ if } k\in   \{2,4, \cdots, 2n\}, \\ 0 & \text{ otherwise.} \end{array} \right.$
\end{itemize}
As in Example \ref{osp_4n pm1},  we get $\mathfrak{osp}(M|2n)$ via $
 \mathcal{J}_{M|2n}= \sum_{k\in I} (-1)^{\delta_k} \bold{e}_{k\, k'}$ for $k'=|I|+1-k.$
We can also define $F_{ij}$ and $\mathcal{B}$ through \eqref{F} and \eqref{basis of osp}. 
Then, $\{F_{ij}|(i,j)\in \mathcal{B}\}$ is a basis of $\mathfrak{osp}(M|2n)$.  Finally, the element 
\begin{equation}
f=  F_{\frac{|I|}{2}+1\, \frac{|I|}{2}-1}+\sum_{i=1}^{\frac{|I|}{2}-1} F_{i+1\, i}
\end{equation}
is an odd principal nilpotent element.
\end{example}

\subsection{Supersymmetric Poisson vertex algebras }  \

Let $\mathcal{R}$ be a vector superspace with an odd operator $D: \mathcal{R} \to \mathcal{R}$, i.e. $\mathcal{R}$ is a $\mathbb{C}[D]$-module. Consider a super non-commutative associative algebra $\mathbb{C}[\chi]$ generated by an odd indeterminate $\chi$. We define the $\mathbb{C}[D]$-module structure of 
$\mathbb{C}[\chi] \otimes \mathcal{R}$  by 
\begin{equation} \label{chi and D}
 D(\chi \otimes  R) = -\chi \otimes D (R) -2\chi^2 \otimes R
 \end{equation}
for $R\in \mathcal{R}$. Note that we usually write elements in $\mathbb{C}[\chi] \otimes \mathcal{R}$ without tensor product. In this way, \eqref{chi and D} can be simply written as follows:
\begin{equation} \label{chi and D, simple}
 D\chi+\chi D = -2\chi^2.
 \end{equation}

A bilinear bracket 
\[ \, [ \ {}_\chi {}\ ] : \mathcal{R} \times \mathcal{R} \to \mathbb{C}[\chi] \otimes \mathcal{R}\]
on the $\mathbb{C}[D]$-module $\mathcal{R}$ is called a $\chi$-bracket if it satisfies the {\it sesquilinearity}
property:
\begin{equation} \label{sesqui}
 \, [ D a\, {}_\chi \,  b\, ] = -\chi [\, a\, {}_\chi\,  b \, ], \quad [\, a \,  {}_\chi \, Db \,  ] = (-1)^{p(a)+1} (D+\chi) [\, a\, {}_\chi \,  b\, ] .
 \end{equation}

\begin{definition} \label{Def:SUSY LCA}
Let $\mathcal{R}$ be a  $\mathbb{C}[D]$-module with a $\chi$-bracket $[ \ {}_\chi \ ]$. If the bracket satisfies the following properties, then we call $\mathcal{R}$ a ($N_k=1$) {\it supersymmetric (SUSY) Lie conformal algebra (LCA)}: for any $a,b,c\in \mathcal{R}$, 
\begin{itemize}
\item (skewsymmetry) $[a{}_\chi b] = (-1)^{p(a)p(b)} [b{}_{-\chi-D} a]$, 
\item (Jacobi identity) \\
 $\quad{} [a{}_\chi [b{}_\gamma c]]= (-1)^{p(a)+1} [[a{}_\chi b]_{\chi+ \gamma}  c]+ (-1)^{(p(a)+1)(p(b)+1)}[b{}_\gamma [a_\chi c]] \in \mathbb{C}[\chi, \gamma] \otimes \mathcal{R}$, 
\end{itemize}
where $\mathbb{C}[\chi, \gamma]$ is a super noncommutative associative algebra generated by  odd indeterminates $\chi$ and $\gamma$ such that $\chi\gamma = -\gamma \chi$.  In addition, $ \mathbb{C}[\chi, \gamma] \otimes \mathcal{R}$ is a $\mathbb{C}[D]$-module  via $\gamma D+ D \gamma =-2 \gamma^2$ and \eqref{chi and D, simple}.
\end{definition}

Let $\mathcal{R}$ be a SUSY LCA and  denote 
\begin{equation} 
\  \   [ \, a\, {}_\chi \, b\, ] = \sum_{n \in \mathbb{Z}_{\geq 0}} \chi^n a_{(n)}b 
\end{equation}
for $a,b\in \mathcal{R}$. Then the skewsymmetry in Definition \ref{Def:SUSY LCA} can be written as 
\begin{equation} \label{skewsymmetry-2}
\  \  \sum_{n\in \mathbb{Z}_{\geq 0}} \chi^n a_{(n)}b = (-1)^{p(a)p(b)} \sum_{n\in \mathbb{Z}_{\geq 0}} (-\chi-D)^n b_{(n)}a,
\end{equation}
where the RHS  is computed by \eqref{chi and D, simple}. 
The Jacobi identity in Definition \ref{Def:SUSY LCA} can be computed by 
\begin{equation} \label{Jacobi-2}
\begin{aligned}
&\  \   [ \, a\ {}_\chi \,  \gamma^n b\, ] = (-1)^{n (p(a)+1)} \gamma^n [\, a\, {}_\chi\,  b\, ],\\
&   \  \   [  \,  \chi^n a\,  {}_{\chi+\gamma} \, b \, ] = (-1)^{n } \chi^n [\, a\, {}_{\chi+\gamma} \,  b\, ] = (-1)^{n } \chi^n \sum_{m\in \mathbb{Z}_{\geq 0}} (\chi+\gamma)^m a_{(m)} b.
\end{aligned}
\end{equation} 

\pagebreak

\begin{definition}\  \label{Def:SUSY PVA}
\begin{enumerate}
\item Let $A$ be a superalgebra and $D:A\to A$ be a linear operator of parity $p(D)$. If $D$ is a derivation, that is, 
\[ D(ab)= D(a) b+ (-1)^{p(D)p(a)} a D(b)\]
for $a,b\in A$,  then $A$ is called a {\it differential algebra} associated with  $D$.

\item Let a $\mathbb{C}[D]$-module $\mathcal{V}$ be a SUSY LCA endowed with a $\chi$-bracket $\{  \ {}_\chi \ {}\}$. If $\mathcal{V}$ is a unital supercommutative associative algebra such that 
\begin{itemize}
\item $\mathcal{V}$ is a differential algebra with respect to the odd derivation $D$,
\item (Leibniz rule) $\{\, a\, {}_\chi \, bc\, \}  = (-1)^{(p(a)+1)p(b)} b\{\, a\, {}_\chi \, c\, \} + \{\, a \, {}_\chi \, b\, \}c$ for $a,b,c\in \mathcal{V}$.
\end{itemize}
Then, $\mathcal{V}$ is called a {\it SUSY Poisson vertex algebra (PVA)}. 
\end{enumerate}
\end{definition}

The Leibniz rule in Definition \ref{Def:SUSY PVA} and the skewsymmetry of SUSY PVA imply the right Leibniz rule 
\begin{equation}
\{\, ab\, {}_\chi \, c\, \} = (-1)^{p(c)p(b)}\{\, a\, {}_{\chi+D} \, c\, \}_{\to} b+(-1)^{p(a)(p(b)+p(c))}\{\, b\, {}_{\chi+D} \, c\, \}_{\to} a,
\end{equation}
where $\{\, a\, {}_{\chi+D} \, c\, \}_{\to} b= \sum_{n \in \mathbb{Z}_{\geq 0}} a_{(n)}c \,(\chi+D)^n b$. Hence, if a SUSY PVA $\mathcal{V}$ is generated by a set $\mathcal{B}$ as a differential algebra, then the $\chi$-bracket on $\mathcal{B}$ completely determines the $\chi$-bracket on $\mathcal{V}$ by the Leibniz rule and the sesquilinearity. Moreover, using the following theorem, one can derive a SUSY PVA from a given SUSY LCA.

\begin{theorem} \label{Thm:fund_PVA_1}\cite{CS21}
Let $\mathcal{R}$  be a SUSY LCA . Then the supersymmetric algebra 
\[S(\mathcal{R}):=S(\mathcal{R}_{\bar{0}})\otimes \bigwedge(\mathcal{R}_{\bar{1}})\] endowed with the $\chi$-bracket induced from the $\chi$-bracket of $\mathcal{R}$ and the Leibniz rule is a SUSY PVA. 
\end{theorem}

In addition, the following theorem is useful to check axioms of SUSY LCAs. 

\begin{theorem}\cite{CS21}\label{Thm:fund_PVA_2}
 Let $V$ be a vector superspace with a basis $\mathcal{B}=\{v_1, \cdots, v_m\}$ and let $\mathcal{R}= \mathbb{C}[D]\otimes V$ be a $\mathbb{C}[D]$-module endowed with a $\chi$-bracket $[\ {}_\chi \ ]$. If the $\chi$-bracket $[\ {}_\chi \ ] |_{\mathcal{B}\times \mathcal{B}}$ restricted to $\mathcal{B}$ satisfies the  skewsymmetry and Jacobi identity then $\mathcal{R}$ is a SUSY LCA with the bracket $[\ {}_\chi \ ]$. 
\end{theorem}

\begin{example} \label{Ex:affine}
Let $\g$ be a finite simple Lie superalgebra with a  supersymmetric invariant bilinear form $(\ | \ )$. Consider the vector superspace $\bar{\g}:= \{ \bar{a}| a\in \g\}$ where the parity $p(\bar{a})$  of a homogeneous element $a\in \g$ is defined by  $p(\bar{a})\equiv 1+ p(a) \Mod2$. Let us define the $\chi$-bracket on $\mathcal{R}(\bar{\g}):= \mathbb{C}[D] \otimes \bar{\g}\oplus K \mathbb{C}$ by 
\begin{equation} \label{Eqn:affine}
 [ \bar{a}{}_\chi \bar{b}] = (-1)^{p(a) p(\bar{b})} \overline{[a,b]} + \chi K (a|b), \quad [K {}_\chi \bar{a}]= 0
\end{equation}
for $a,b\in \g$. One can check that the $\chi$-bracket on $\bar{\g}\oplus \mathbb{C} K$ satisfies the skewsymmetry and Jacobi identity, so $\mathcal{R}(\bar{\g})$ is a SUSY LCA, called the {\it SUSY current LCA.}  

 By Theorem \ref{Thm:fund_PVA_2}, the supersymmetric algebra
$ S(\mathcal{R}(\bar{\g}))$
 is a SUSY PVA endowed with the $\chi$-bracket induced from the bracket of $\mathcal{R}(\bar{\g})$. Obviously, 
\begin{equation} \label{affine}
\mathcal{V}^k(\bar{\g}):= S(\mathcal{R}(\bar{\g})) / (K-k)S(\mathcal{R}(\bar{\g}))
\end{equation}
is a quotient SUSY PVA for any $k\in \mathbb{C}$. This algebra is called the {\it SUSY affine PVA} associated with $\g$ and $k$. 
\end{example}

\begin{remark}
In \cite{HK06}, the SUSY affine vertex algebra $\widetilde{\mathcal{V}}^k(\bar{\g})$ associated with a Lie superalgebra $\g$ is endowed with the $\chi$-bracket 
\begin{equation} \label{Eqn:affine_2}
 \{ \bar{a}{}_\chi \bar{b} \} = (-1)^{p(a)} (\overline{[a,b]} + k\chi\, (a|b)) \ \text{ for } \  a,b\in \g. 
\end{equation}
One can show that the SUSY PVA defined via  \eqref{Eqn:affine_2} is isomorphic to $\mathcal{V}^k (\bar{\g})$ by considering the map $\bar{a} \mapsto  \mathfrak{i}^{p(a)} \bar{a}$ for a homogeneous element $a\in \g$ and the imaginary number $\mathfrak{i}\in \mathbb{C}$. 
\end{remark}

\begin{remark} \label{rem:classical limit} \cite{MRS20}
The SUSY PVA $\mathcal{V}^k(\bar{\g})$ is obtained from the SUSY current LCA $\mathcal{R}(\bar{\g})$. In a similar manner, a SUSY vertex algebra $V^k(\bar{\g})$ called SUSY universal affine vertex algebra associated to $\mathcal{R}(\bar{\g})$ can be constructed  via Wick formula. One can relate $\mathcal{V}^k(\bar{\g})$ and $V^k(\bar{\g})$ by the process called quasi-classical limit. In general,  the universal SUSY vertex algebra $V(\mathcal{R})$ can be constructed for any given SUSY LCA $\mathcal{R}$. The SUSY PVA $S(\mathbb{C}[D]\otimes \mathcal{R})$ can be understood as a quasi-classical limit of $V(\mathcal{R})$.  
\end{remark}

\section{SUSY classical $W$-algebras\label{sec:SUSYW}}

In this section, let $\g$ be a finite simple Lie superalgebra with a principal subalgebra $\mathfrak{s}\simeq \mathfrak{osp}(1|2).$ The subspace $\mathfrak{s}$ is spanned by five elements: $F, h=2x, E$ which consists of a $\mathfrak{sl}_2$-triple and two odd elements $e,f$ such that $[e,e]=2E$ and $[f,f]=-2F$. Then, by the $\mathfrak{sl}_2$ representation theory, 
\begin{equation} \label{grading}
\  \  \g= \bigoplus_{i\in \frac{\mathbb{Z}}{2}} \g_i,
\end{equation}
where $\g_i\subset \g$ is the eigenspace associated with $\text{ad}\, x$.
We assume that $\g$ has a supersymmetric nondegenerate invariant bilinear form $(\, | \, )$ such that 
\[ (E|F)= 2(x|x)=1. \]

 Note that 
\[ F\in \g_{-1}, \ f\in \g_{-1/2}, \, h \in \g_0, \ e \in \g_{1/2}, \, E\in \g_1,\]
and we denote 
\begin{equation}\label{notation:subalgebra}
\  \  \mathfrak{n}:= \bigoplus_{i>0} \g_i, \quad \mathfrak{n}_-:= \bigoplus_{i<0} \g_i, \quad \mathfrak{p}:= \bigoplus_{i\leq 0} \g_i, \quad \mathfrak{p}_+:= \bigoplus_{i\geq 0} \g_i.
\end{equation}
Then, since the $\mathfrak{sl}_2$-representation theory tells that  $\text{ad}\, F |_{\mathfrak{n}}$ is injective and  $\text{ad}\, F |_{\g_{\leq 1/2}}$ is surjective, one can see that $ \text{ad}\, f|_{\mathfrak{n}}$ is injective and  $\text{ad} \, f|_{\mathfrak{p}}$ is surjective.\\

\subsection{SUSY classical $W$-algebras} \ 

In this section, we introduce a construction of SUSY classical $W$-algebras. See \cite{MRS20} for  the quantum case.

Recall that $\mathcal{V}^k(\bar{\g})$ is the SUSY  affine PVA associated with $\g$ introduced in Example \ref{Ex:affine} and let $\mathcal{V}^k(\bar{\mathfrak{p}})=S(\mathbb{C}[D]\otimes \bar{\mathfrak{p}})$ be the SUSY Poisson vertex subalgebra of $\mathcal{V}^k(\bar{\g})$. The SUSY classical $W$-algebras associated with $\g$ and $f$ are defined as follows.

\begin{tdefinition}\cite{Suh} \label{prop:another definition}
 Let $\mathfrak{I}$ be the differential algebra ideal of $\mathcal{V}^k(\bar{\g})$ generated by  $\{ \bar{n} -(f|n) | n\in \n\}$.
Consider 
\begin{equation}\label{second_W}
\mathcal{W}(\bar{\g}, f, k):= \{ w\in \mathcal{V}^k(\bar{\mathfrak{p}}) |  \{ \bar{n}{}_\chi w\} \in \mathbb{C}[\chi]\otimes \mathfrak{I} \text{ for any } n\in \bar{\n} \}
 \end{equation}
 where the $\chi$-bracket in \eqref{second_W} is the bracket on $\mathcal{V}^k (\bar{\g})$. Then $\mathcal{W}(\bar{\g}, f, k)$ is a well-defined SUSY PVA with the $\chi$-bracket induced from  $\mathcal{V}^k (\bar{\g})$ and it is called the SUSY classical $W$-algebra associated with $\g$ and $f$.
\end{tdefinition}
 
 In order to describe a generating set of the $W$-algebra, let us define the  $\frac{\mathbb{Z}}{2}$-grading $\Delta$ on $\mathcal{V}^k (\bar{\mathfrak{p}})$ by 
 \begin{equation} \label{Eqn:conformal weight}
  \  \  \Delta_{\bar{a}} = \frac{1}{2}-j_a, \quad \Delta_{AB}= \Delta_A +\Delta_B, \quad \Delta_{D}= \frac{1}{2}
  \end{equation}
for $a\in \g_{j_a}$ and $A,B \in \mathcal{V}^k (\bar{\mathfrak{p}})$.

\begin{proposition}\cite{MRS20}\label{Prop:property_generator}
Let $ \g^{f}:= \ker\, ( \textup{ad}\, f )$, with a basis $\{v_i| i=1, \cdots, n_f\}\subset \mathfrak{p}$ 
 and let $v_i \in \g_{j_i}$ for some $j_i\in \frac{\mathbb{Z}_{\leq 0}}{2}$. 
As a differential algebra, 
\[\mathcal{W}(\bar{\g}, f,k)= \mathbb{C}[D^m w_i:= w_i^{(m)} | m\in \mathbb{Z}_{\geq 0}, \, i=1, 2, \cdots, n_f]\]
 where $w_1, \cdots, w_{n_f}$ satisfy the following properties:
\begin{enumerate}[(i)]
\item $\Delta_{w_i}= \Delta_{\bar{v}_i}$ for $i=1, \cdots, n_f$, \label{Prop: generator conf wt}
\item $w_i\in \bar{v}_i+ \bigoplus_{m\geq 2} S^m (\mathbb{C}[D] \otimes \bar{\mathfrak{p}}) + \bigoplus_{m\geq 1} D^m(\bar{\mathfrak{p}})$. In other words, 
 the  linear part  of $w_i$ without $D$    is $\bar{v}_i$.  
\end{enumerate}
Moreover, a set $\{w_i|i=1, \cdots, n_f\} \subset \mathcal{W}(\bar{\g},f,k)$ satisfying (i) and (ii)  freely generates $\mathcal{W}(\bar{\g},f,k)$. 
\end{proposition}

\begin{remark} \label{Rem:notation}In the following sections, to simplify the presentation, we will note:
\[\mathcal{W}(\bar{\g}, f):= \mathcal{W}(\bar{\g}, f, 1), \quad \mathcal{V}(\bar{\g}):=  \mathcal{V}^1(\bar{\g}), \quad  \mathcal{V}(\bar{\mathfrak{p}}):=  \mathcal{V}^1(\bar{\mathfrak{p}}). \]
\end{remark}

\subsection{$\chi$-adjoint action}\ 

Let $\mathcal{P}$ be a SUSY PVA with an odd derivation $D$.
We consider a $\mathbb{C}(\!(D^{-1})\!)$-module $\mathcal{P}(\!(D^{-1})\!)$  such that 
\begin{equation*}
 \  \  D \big( a D^n\big)= a'D^n+(-1)^{p(a)} a D^{n+1}
\end{equation*}
and 
\begin{equation*}
\begin{aligned}
 \  \  D^{-1} \big( a D^n\big)& = 
 (-1)^{p(a)} a D^{n-1} +a' D^{n-2} -(-1)^{p(a)} a^{(2)} D^{n-3} -a^{(3)} D^{n-4} + \cdots 
 \\
 & = \sum_{m\geq 0} (-1)^m\left( (-1)^{p(a)} a^{(2m)}D^{-1-2m} +a^{(2m+1)}D^{-2-2m} \right)D^n
\end{aligned}
\end{equation*}
where $n\in \mathbb{Z}$, $a\in \mathcal{P}$,  $a':=D(a)$ and $a^{(m)}= D^m(a)$ for $m\geq 0$.

\begin{definition} \label{Def-adjoint}
The {\it $\chi$-adjoint action} of  $\mathcal{P}$ in $\mathcal{P}(\!(D^{-1})\!)$  is defined by 
\begin{equation}
\  \  \text{ad}_\chi P : \mathcal{P}(\!(D^{-1})\!) \to \mathbb{C}[\![\chi]\!]\otimes \mathcal{P}(\!(D^{-1})\!), \quad \sum_{n\in \mathbb{Z}} a_n D^n \mapsto \sum_{n\in \mathbb{Z}} \{P {}_\chi a_n\}D^n
\end{equation}
for $P\in \mathcal{P}$.
\end{definition}

Note that we also consider $\mathbb{C}[\![\chi]\!]\otimes \mathcal{P}(\!(D^{-1})\!)$ as a  $\mathbb{C}(\!(D^{-1})\!)$-module via 
\begin{equation*}
D\chi=-\chi D-2 \chi^2, \quad  D^{-1} \chi = -\chi D^{-1} -2\chi^2 D^{-2}.
\end{equation*}
Since 
$\left[\sum_{n\geq 0} \chi^{n} D^{-n-1} \right](\chi+D)=(\chi+D)\left[\sum_{n\geq 0} \chi^{n} D^{-n-1} \right]$ is the identity  on $\mathbb{C}[\![\chi]\!]\otimes \mathcal{P}(\!(D^{-1})\!)$, we  denote 
\[(\chi+D)^{-1}:= \sum_{n\geq 0} \chi^{n} D^{-n-1}.\]

\begin{lemma} \label{Lemma:nonlocal}
For any $P\in \mathcal{P}$ and $A(D) \in \mathcal{P}(\!(D^{-1})\!)$, we have 
\begin{enumerate}
\item $  \textup{ad}_\chi P\, \big( D\,A(D)\big)= (-1)^{p(P)+1}(\chi+D) \textup{ad}_\chi P\, \big(  A(D)\big),$
\item $  \textup{ad}_\chi P\, \big( D^{-1} A(D)\big)= (-1)^{p(P)+1}(\chi+D)^{-1}   \textup{ad}_\chi P\, \big(  A(D)\big).$
\end{enumerate}
\end{lemma}
\begin{proof}
The equation (1) is a consequence of  sesquilinearity in a SUSY PVA. Apply (1) to $D^{-1}A(D)\in \mathcal{P}(\!(D^{-1})\!)$. Then \[(\chi+D)\, \text{ad}_\chi P\, \big( D^{-1} A(D)\big)=(-1)^{p(P)+1} \text{ad}_\chi P\, \big(DD^{-1}A(D)\big)=(-1)^{p(P)+1} \text{ad}_\chi P\, \big(A(D)\big).\] Since $(\chi+D)^{-1}$ is the left inverse of $(\chi+D)$, we get the equality (2).
\end{proof}

\begin{proposition} \label{Prop:Leibniz, P(D)}
 For $P\in \mathcal{P}$ and $A(D), B(D) \in \mathcal{P}(\!(D^{-1})\!)$, we have the following property:
\begin{equation} \label{Eqn:prop+leibP(D)}
\begin{aligned}
& \textup{ad}_\chi P\,  \big( A(D)B(D) \big) \\
&=  (-1)^{p(A(D))(p(P)+1)}A(D+\chi)\  \textup{ad}_\chi P\,  \big(B(D)
\big)+ \textup{ad}_\chi P\,  \big(A(D) \big)B(D).
\end{aligned}
\end{equation}
\end{proposition}

\begin{proof}
Let  $a,b\in \mathcal{P}$ and $m,n\in \mathbb{Z}$. It is enough to show \eqref{Eqn:prop+leibP(D)} for $A(D)= aD^m$ and $B(D)= bD^n$ . If $m=0$, it directly follows from the Leibniz rule of $\chi$-bracket on $\mathcal{P}$.
  If $m=1$, then 
\begin{equation}
\begin{aligned}
&  \text{ad}_\chi P\,  \big( aD\, bD^n \big)  =  \text{ad}_\chi P\, (ab'D^n +(-1)^{p(b)}ab D^{n+1})\\
& = \text{ad}_\chi  P\, (a) \,(b' D^n + (-1)^{p(b)}b D^{n+1})\\
& + (-1)^{p(a) (p(P)+1)} a   \big(\text{ad}_\chi P \, (b') \, D^n +(-1)^{p(b)} \text{ad}_\chi P\, (b)\, D^{n+1}\big)\\
& =  \text{ad}_\chi P\, (aD)\, bD^n+(-1)^{(p(a)+1) (p(P)+1)}a (D+\chi) \text{ad}_\chi P\,(bD^n).
\end{aligned}
\end{equation}
Inductively, one can prove the proposition for $m\geq 1$. 
If $m<0$, similar arguments work by Lemma \ref{Lemma:nonlocal}.
\end{proof}

\begin{corollary} \label{Cor1}
Let $P\in \mathcal{P}$ and $A_1(D), A_2(D),\cdots, A_l(D) \in \mathcal{P}(\!(D^{-1})\!)$. Then 
\begin{equation*}
\begin{aligned}
 & \textup{ad}_\chi  P\,  \big( A_1(D)A_2(D)  \cdots A_l (D)\big)  \\
 & = \  \  \sum_{i=1}^{l}  (-1)^{\sum_{k=1}^{i-1}p(A_k(D))\left(p(P)+1\right)}\left[ \prod_{j=1}^{i-1} \left( A_j(D+\chi)\right) \textup{ad}_\chi  P (A_i(D))\prod_{j'=i+1}^{l} \left( A_{j'}(D)\right) \right].
 \end{aligned}
\end{equation*}
\end{corollary}

 Recall the notations in Theorem-Definition \ref{prop:another definition} and Remark \ref{Rem:notation}.  Consider the differential algebra homomorphism 
\[ \pi: \mathcal{V}(\bar{\g})(\!(D^{-1})\!)\to  \mathcal{V}(\bar{\mathfrak{p}})(\!(D^{-1})\!) \]
defined by $\pi(\bar{a}D^n)= \bar{a}D^n$ for $a\in \mathfrak{p}$ and $\pi(\bar{a}D^n)= (f|a)D^n$ for $a \in \n$. One can canonically extend the map $\pi$ to 
\[ \widetilde{\pi}: \mathbb{C}[\chi]\otimes \mathcal{V}(\bar{\g})(\!(D^{-1})\!)\to  \mathbb{C}[\chi] \otimes  \mathcal{V}(\bar{\mathfrak{p}})(\!(D^{-1})\!) .\]

In the following sections, Corollary \ref{Cor2} is used to find generators of $\mathcal{W}(\bar{\g}, f)$. 

\begin{corollary}\label{Cor2}
Let $P(D)= \sum_{n\in \mathbb{Z}} p_n D^n\in \mathcal{V}(\mathfrak{\bar{p}})(\!(D^{-1})\!)\subset  \mathcal{V}(\mathfrak{\bar{g}})(\!(D^{-1})\!)$ for $p_n \in\mathcal{V}(\mathfrak{\bar{p}})$. All coefficients $p_n$ are in  $\mathcal{W}(\bar{\g}, f)$ if and only if 
\[ \textup{ad}_\chi^{\, \mathfrak{J}} a\, (P(D)):= \widetilde{\pi} \circ \textup{ad}_\chi a\, (P(D))=0\]
for any $a \in \n$.

\end{corollary}

\section{Generators of $\mathcal{W}(\overline{\frak{sl}}(n\pm 1|n), f)$\label{sec:glmn}}

In this section, we find a generating set of the SUSY classical $W$-algebra associated with $\mathfrak{gl}(n\pm1|n)$ or $\mathfrak{sl}(n\pm 1|n)$ and the odd principal nilpotent element $f$. 
 In \cite{MRS20}, one can find the  quantum analogue of   $\mathfrak{gl}(n+1|n)$ case.

\subsection{$\mathcal{W}(\overline{\frak{gl}}(n+1|n),f)$} \label{Subsec:gl}\

Let us use the matrix presentation of $\mathfrak{gl}(n+1|n)$ and $f$ in Example \ref{gl(n+1)}.
Then the $\frac{\mathbb{Z}}{2}$-grading \eqref{grading} on $\g=\frak{gl}(n+1|n)$ and the subalgebras $\mathfrak{p}$ and $\n$ in \eqref{notation:subalgebra} are determined.
\vskip 2mm
 Let $\mathcal{V}(\bar{\mathfrak{p}}):= S(\mathbb{C}[D]\otimes \bar{\mathfrak{p}})$ and $\bold{e}^{ij}:= (-1)^{p(i)} \bold{e}_{ij}=(-1)^{i+1}\bold{e}_{ij}$.
Consider the $({2n+1}) \times (2n+1)$ matrix 
\begin{equation*} \label{A_gl}
\begin{aligned}
\mathcal{A} = \tiny \left( \begin{array}{ccccccccc} D+\bar{\bold{e}}^{11} & \bar{\bold{e}}^{21}&  \cdots & \bar{\bold{e}}^{n+1\, 1} &  \bar{\bold{e}}^{n+2\, 1}& \cdots& \bar{\bold{e}}^{2n\, 1}& \bar{\bold{e}}^{2n+1\, 1}  \\
 -1& D+\bar{\bold{e}}^{22}  & \cdots &  \bar{\bold{e}}^{n+1\, 2} &  \bar{\bold{e}}^{n+2\, 2}& \cdots &\bar{\bold{e}}^{2n\, 2}& \bar{\bold{e}}^{2n+1\, 2}\\
 \vdots & \vdots&&\vdots&\vdots&&\vdots&\vdots \\
 0 & 0 & \cdots & D+ \bar{\bold{e}}^{n+1\, n+1} & \bar{\bold{e}}^{n+2\, n+1} &  \cdots &\bar{\bold{e}}^{2n\, n+1}& \bar{\bold{e}}^{2n+1\, n+1}\\
  0 & 0& \cdots &-1 & D+ \bar{\bold{e}}^{n+2\, n+2} &  \cdots & \bar{\bold{e}}^{2n\, n+2}& \bar{\bold{e}}^{2n+1\, n+2}\\
   \vdots & \vdots&&\vdots&\vdots&&\vdots&\vdots \\
     0 & 0& \cdots &0 &0&  \cdots &D+\bar{\bold{e}}^{2n\, 2n}&\bar{\bold{e}}^{2n+1\, 2n}\\
          0 & 0& \cdots &0 &0&  \cdots & -1 &D+\bar{\bold{e}}^{2n+1\, 2n+1}\\
   \end{array}\right).
\end{aligned}
\end{equation*}
In other words, 
 \begin{equation} \label{A_general}
 \mathcal{A}= \sum_{i\in I, \, i \leq j} \bold{e}_{ij}\otimes  (\delta_{ij} D+\bar{\bold{e}}^{ji}) -f\otimes 1 \in \g \otimes \mathcal{V}(\bar{\mathfrak{p}})[D] .
\end{equation}
Then the row determinant of $\mathcal{A}$ is 
\begin{equation} \label{row_det:gl(n+1|n)}
rdet(\mathcal{A}):=\sum_{N=0}^{2n} \sum_{\substack{0=i_0 < i_1 < \cdots \\ \cdots < i_{N+1}= 2n+1}} \mathcal{A}_{i_0+1, i_1} \mathcal{A}_{i_1+1, i_2} \cdots \mathcal{A}_{i_N+1, i_{N+1}},
\end{equation}
where $\mathcal{A}_{ij}$ denotes the $ij-$entry of $\mathcal{A}$.

By direct computations, we get the following lemma.
\begin{lemma}  \label{Lemma:gl}
For $i, k, l\in I$ such that $i\leq 2n$, we have 
\begin{enumerate}
\item$ \{ \bar{\bold{e}}_{i\, i+1}{}_\chi \bar{\bold{e}}^{kl}\}= \delta_{i+1, k} (-1)^l \bar{\bold{e}}_{il}-(-1)^{k} \delta_{i,l} \bar{\bold{e}}_{k, i+1} - \chi \delta_{i,l}\delta_{i+1,k}$,
\item  $\bar{\bold{e}}_{i,i+1}= (-1)^{i}$ and $\bar{\bold{e}}^{i,i+1}= -1$ in $\mathfrak{J}$, where $\mathfrak{J}$ is in Theorem-Definition \ref{prop:another definition}.
\end{enumerate}
\end{lemma}

\begin{theorem} \label{Thm:gl(n+1|n)}
Let $\eqref{row_det:gl(n+1|n)} = \sum_{i=0}^{2n+1} w_{2n+1-i} D^i$ for $w_i\in \mathcal{V}(\bar{\mathfrak{p}})$.  Then 
\[ \mathcal{W}(\bar{\g}, f)= \mathbb{C}[w_i^{(m)} | i\in I, \  m \in \mathbb{Z}_{\geq 0}].\]
\end{theorem}
\begin{proof}
In order to show that $w_i^{(m)} \in \mathcal{W}(\bar{\g}, f)$ for all $i\in I$ and $m\in \mathbb{Z}_{\geq 0}$, 
it is enough to check that $\text{ad}_\chi^{\mathfrak{J}} \, \bold{\bar{e}}_{i, i+1}(rdet (\mathcal{A}))=0$. Observe that
\begin{equation*}
\begin{aligned}
& \text{ad}_\chi^{\mathfrak{J}} \, \bold{\bar{e}}_{i, i+1} (\mathcal{A}_{i k})= \mathcal{A}_{i+1,k}, \quad
 \text{ad}_\chi^{\mathfrak{J}} \, \bold{\bar{e}}_{i, i+1} (\mathcal{A}_{i i} \mathcal{A}_{i+1, k})= -\mathcal{A}_{i+1, k} \ \text{ for } \  k>i+1 ,  \\
& \text{ad}_\chi^{\mathfrak{J}} \, \bold{\bar{e}}_{i, i+1} (\mathcal{A}_{l, i+1})= (-1)^{l+i+1} \mathcal{A}_{li}, \quad
\text{ad}_\chi^{\mathfrak{J}} \, \bold{\bar{e}}_{i, i+1} (\mathcal{A}_{li} \mathcal{A}_{i+1, i+1})= (-1)^{l+i}\mathcal{A}_{li} \ \text{ for } \  l<i.
\end{aligned}
\end{equation*}
In addition, 
\begin{equation}
\begin{aligned}
&  \text{ad}_\chi^{\mathfrak{J}} \, \bold{\bar{e}}_{i, i+1}(\mathcal{A}_{i, i+1}) = (-1)^i\bar{\bold{e}}_{ii}+(-1)^i \bar{\bold{e}}_{i+1,i+1} -\chi,\\
& \text{ad}_\chi^{\mathfrak{J}} \, \bold{\bar{e}}_{i, i+1}(\mathcal{A}_{i, i}\mathcal{A}_{i+1, i+1}) = -(-1)^i\bar{\bold{e}}_{ii}-(-1)^i \bar{\bold{e}}_{i+1,i+1}+\chi.
\end{aligned}
\end{equation}
By Corollary \ref{Cor2}, we can see $w_i \in  \mathcal{W}(\bar{\g}, f)$ for $i\in I$.

In order to see $\{w_i|i\in I\}$ freely generates $\mathcal{W}(\bar{\g}, f)$, 
observe that $\Delta_{w_i}=\frac{i}{2}$. Moreover,  by direct computations, one can check that there is a basis $\{v_i |i\in I\}$ of $\g^f$ which let the set 
$\{w_1, \cdots, w_{2n+1}\}$ satisfy (i) and (ii) in Proposition \ref{Prop:property_generator}.
\end{proof}

\begin{remark}
In order to find generators of $\mathcal{W}(\bar{\g}, f,k)$, one can just replace $D$ by $kD$ in the matrix $\mathcal{A}$ and compute the row determinant.
\end{remark}


\subsection{$\mathcal{W}(\overline{\mathfrak{sl}}(n+1|n), f)$} \ 

Let $\g= \mathfrak{sl}(n+1|n)$. \label{Sec:sl}
 Recall  $\g=\{a\in \mathfrak{gl}(n+1|n)| \textit{str}(a)=0\}$. The odd principal nilpotent element $f= \sum_{i=1}^{2n} \bold{e}_{n+1, n}$ in $\g$  induces the $\frac{\mathbb{Z}}{2}$-grading and 
\begin{equation}
\begin{aligned}
&\mathfrak{n}:= \g_{> 0} = \mathfrak{gl}(n+1|n)_{> 0},\\
& \mathfrak{p}:= \g_{\leq 0} \text{ and }  {\mathfrak{q}}:=\mathfrak{p} \oplus \mathbb{C} I_{(n+1|n)} = \mathfrak{gl}(n+1|n)_{\leq 0}.
\end{aligned}
\end{equation}
where $I_{(n+1|n)}= \sum_{i=1}^{2n+1} \bold{e}_{ii}$.
Consider the linear map
$
 \mathfrak{gl}(n+1|n)  \to \g
$
defined by $A\mapsto A-\text{str}(A) I_{(n+1|n)}$  and let 
\begin{equation} \label{pi_V}
\pi_\mathfrak{sl}: \mathcal{V}(\bar{\mathfrak{q}}):= S(\mathbb{C}[D] \otimes \bar{\mathfrak{q}})\to  \mathcal{V}(\bar{\mathfrak{p}}):=S(\mathbb{C}[D] \otimes \bar{\mathfrak{p}})
\end{equation}
be the induced differential algebra homomorphism.

\begin{theorem}\label{Thm:sl(n+1|n)}
Let $w_1, w_2, \cdots, w_{2n+1} \in \mathcal{V}(\bar{\mathfrak{q}})$ be the elements in Theorem~\ref{Thm:gl(n+1|n)}. Then 
\[
\mathcal{W}(\bar{\g}, f)= \mathbb{C}[\pi_{\mathfrak{sl}}(w_i^{(m)}) | i=2, \cdots, 2n+1, \  m \in \mathbb{Z}_{\geq 0}] \subset \mathcal{V}(\bar{\mathfrak{p}}).
\]
\end{theorem}
\begin{proof}
Since $I_{(n+1|n)}$ is a central element in $\mathfrak{gl}(n+1|n)$,  we have $\pi_{\mathfrak{sl}}(w_i) \in \mathcal{W}(\bar{\g}, f)$ for $i\in I$. In addition, since $\pi_{\mathfrak{sl}}(w_1)=0$, we can prove the theorem. 
\end{proof}

\subsection{$\mathcal{W}(\mathfrak{gl}(n-1|n),f)$ and $\mathcal{W}(\mathfrak{sl}(n-1|n),f)$} \label{Sec:gl(n-1|n)}\ 

Let $n\geq 2.$
 It is well known that $\mathfrak{gl}(n-1|n) \simeq \mathfrak{gl}(n|n-1)$ and $\mathfrak{sl}(n-1|n)\simeq \mathfrak{sl}(n|n-1)$. Hence, the results in  Section \ref{Sec:gl(n-1|n)} follow from those in Section \ref{Subsec:gl} and \ref{Sec:sl}.  However, for the use in Section \ref{sec:ospodd}, we provide clear statements.

Recall the matrix presentation of $\mathfrak{gl}(n-1|n)$ and $f$ in Example \ref{gl(n-1)}.
Consider the $(2n-1) \times (2n-1)$ matrix
 \begin{equation} \label{matrix:gl(n-1|n)}
 \mathcal{A}= \sum_{ i,j\in I,\, i \leq j} \bold{e}_{ij}\otimes  (\delta_{ij} D+\bar{\bold{e}}^{ji}) -f\otimes 1
\end{equation}
whose entries are in $\mathcal{V}(\overline{\mathfrak{gl}}(n-1|n))[D]$.

\begin{theorem}\label{Thm:gl(n-1|n), sl(n-1|n)}
For the  matrix $\mathcal{A}$ in \eqref{matrix:gl(n-1|n)}, write its row determinant as $rdet(\mathcal{A})= \sum_{i=0}^{2n-1} w_{2n-1-i} D^i$ for $w_i\in \mathcal{V}(\overline{\mathfrak{gl}}(n-1|n))$. Then
\begin{enumerate}
\item  $\mathcal{W}(\overline{\mathfrak{gl}}(n-1|n), f)= \mathbb{C}[w_i^{(m)} | i\in I, \  m \in \mathbb{Z}_{\geq 0}]$,
\item $\mathcal{W}(\overline{\mathfrak{sl}}(n-1|n), f)= \mathbb{C}[\pi_{\mathfrak{sl}}(w_i^{(m)}) | i=2, \cdots, 2n-1, \  m \in \mathbb{Z}_{\geq 0}]$, where $\pi_{\mathfrak{sl}}: \mathcal{V}(\overline{\mathfrak{gl}}(n-1|n)) \to \mathcal{V}(\overline{\mathfrak{sl}}(n-1|n))$ is the differential algebra homomorphism defined by $\bar{A}\mapsto \bar{A}- \text{str}(A) \bar{I}_{(n-1|n)}$ for  $A\in \g$ and $I_{(n-1|n)}=\sum_{i=1}^{2n-1} \bold{e}_{ii}$.
\end{enumerate}
\end{theorem}
\begin{proof}
The proof of (1) is analogous to the proof of Theorem \ref{Thm:gl(n+1|n)}. Also, as in Theorem \ref{Thm:sl(n+1|n)}, (2) follows from (1) for the same reason.
\end{proof}

\section{Generators of $\mathcal{W}(\overline{\frak{osp}}(2n\pm1|2n), f)$ \label{sec:ospodd}} 
In this section, we find a generating set of the SUSY classical $W$-algebra associated with $\frak{osp}(2n\pm1|2n)$ and its principal nilpotent $f$. We use the matrix presentation of  $\frak{osp}(2n\pm1|2n)$ and
the  odd principal nilpotent element $f$ 
given in Example \ref{osp_4n pm1}.

\subsection{$\mathcal{W}(\overline{\mathfrak{osp}}(2n+1|2n), f)$} \ 

\label{Sec:osp(2n+1|2n)}
 Let $\mathfrak{gl}:= \mathfrak{gl}(2n+1|2n)$ and  $\g= \mathfrak{osp}(2n+1|2n)$. 
Consider  the automorphism $\check{\theta}= -\theta$  of $\mathfrak{gl}$, built from  the automorphism $\theta$ in \eqref{Theta}.

As vector superspaces, we have 
\begin{equation}
\mathfrak{gl}= \mathfrak{gl}^{\theta}\oplus \mathfrak{gl}^{\check{\theta}}=\g \oplus \mathfrak{gl}^{\check{\theta}},
\end{equation}
where   $\mathfrak{gl}^{\check{\theta}}$ is spanned by $\bold{e}_{ij}+\check{\theta}(\bold{e}_{ij})$ for $i,j\in I$. By direct computations, one can check the following lemma.

\begin{lemma}\label{Lem:osp} 
(1)  $[\g, \mathfrak{gl}^{\check{\theta}}] \subset  \mathfrak{gl}^{\check{\theta}}$ \qquad (2) $[ \mathfrak{gl}^{\check{\theta}}, \mathfrak{gl}^{\check{\theta}}] \subset\g$.
\end{lemma}

Consider the $\mathbb{C}[D]$-module
\begin{equation}
R= (\mathbb{C}[D] \otimes \overline{\g})\oplus (\mathbb{C}[D] \otimes \overline{\mathfrak{gl}}^{\check{\theta}})^{\otimes 2}.
\end{equation}
By Lemma \ref{Lem:osp}, the supersymmetric algebra $S(R)$ is a SUSY PVA  
endowed with the $\chi$-bracket induced from the bracket of SUSY affine PVA $\mathcal{V}(\overline{\mathfrak{gl}})$:
\[ [\bar{a}{}_\chi \bar{b}] =(-1)^{p(a)p(b)+p(a)}\overline{[a,b]} + \chi\,(a|b). \]
Then as $\mathbb{C}[D]$-modules,
\begin{equation} \label{decomposition}
S(\mathbb{C}[D]\otimes\overline{\mathfrak{gl}})=S(R) \oplus (\mathbb{C}[D] \otimes \overline{\mathfrak{gl}}^{\check{\theta}}) S(R).
\end{equation}
Consider the composition  
\begin{equation} \label{kappa}
\kappa: S(\mathbb{C}[D]\otimes \overline{\mathfrak{gl}}) \xrightarrow{\pi_1} S(R) \xrightarrow{\psi} S(R) /\mathcal{M} \simeq S(\mathbb{C}[D] \otimes \bar{\g})
\end{equation}
of the projection map $\pi_1$  via \eqref{decomposition} and the canonical quotient map $\psi$
where $\mathcal{M}$ is the differential algebra ideal generated by 
$\  (\mathbb{C}[D] \otimes\overline{\mathfrak{gl}}^{\check{\theta}})^{\otimes 2}.$
Since $\mathcal{M}$ is also a SUSY PVA ideal of $S(R)$, the map $\psi$ is a SUSY PVA homomorphism.

%
%

\begin{lemma} \label{Lemma:W(gl,f_osp)} 
Consider the matrix 
\begin{equation} \label{A:osp}
\  \  \mathcal{A}:= \sum_{\substack{i\leq j \\ \ i,j\in I}} \bold{e}_{ij} \otimes \left(\delta_{ij}D+\bar{\bold{e}}^{ji}\right) -f \otimes 1\in \mathfrak{gl}\otimes\mathcal{V}(\overline{\mathfrak{gl}})[D].
\end{equation}
If $rdet(\mathcal{A})= \sum_{i=0}^{4n+1} w_{4n+1-i} D^i$ for $w_1,w_2,\cdots,w_{4n+1}\in \mathcal{V}(\overline{\frak{gl}}_{\leq 0})$, then $w_1, w_2, \cdots, w_{4n+1}$ freely generate $\mathcal{W}(\overline{\mathfrak{gl}}, f)$. 
\end{lemma}

\begin{proof}

 The row determinant of the matrix $\mathcal{A}$ is
 \begin{equation}
  rdet(\mathcal{A}):=\sum_{N=0}^{4n} \sum_{\substack{0=i_0 < i_1 < \cdots \\ \cdots < i_{N+1}= 4n+1}} (-1)^{n+\sum_{j\in J_N} \delta_j} \mathcal{A}_{i_0+1, i_1} \mathcal{A}_{i_1+1, i_2} \cdots \mathcal{A}_{i_N+1, i_{N+1}},
\end{equation}
where $J_N= \{i_1, i_2, \cdots, i_{N+1} \}$, $\mathcal{A}_{ij}= \delta_{ij}D+ \bar{\bold{e}}^{ji}$ for $j\geq i$ and $\bold{e}^{ji}= (-1)^{j+1} \bold{e}_{ji}$. Let $\mathfrak{J}$ be the differential algebra ideal of $\mathcal{V}(\overline{\mathfrak{gl}})$ generated by $\bar{n}-(f|n)$ for $n\in \mathfrak{gl}_{>0}$.
Observe that
\begin{equation*}
\begin{aligned}
& \text{ad}_\chi^{\mathfrak{J}} \, \bold{\bar{e}}_{i, i+1} (\mathcal{A}_{i k})= \mathcal{A}_{i+1,k}, \quad
 \text{ad}_\chi^{\mathfrak{J}} \, \bold{\bar{e}}_{i, i+1} (\mathcal{A}_{i i} \mathcal{A}_{i+1, k})= -(-1)^{\delta_i}\mathcal{A}_{i+1, k} \ \text{ for } \  k>i+1 ,  \\
& \text{ad}_\chi^{\mathfrak{J}} \, \bold{\bar{e}}_{i, i+1} (\mathcal{A}_{l, i+1})= (-1)^{l+i+1} \mathcal{A}_{li}, \quad
\text{ad}_\chi^{\mathfrak{J}} \, \bold{\bar{e}}_{i, i+1} (\mathcal{A}_{li} \mathcal{A}_{i+1, i+1})= (-1)^{l+i+\delta_i}\mathcal{A}_{li} \ \text{ for } \  l<i.
\end{aligned}
\end{equation*}
In addition, 
\begin{equation*}
\begin{aligned}
&  \text{ad}_\chi^{\mathfrak{J}} \, \bold{\bar{e}}_{i, i+1}(\mathcal{A}_{i, i+1}) = (-1)^i\bar{\bold{e}}_{ii}+(-1)^i \bar{\bold{e}}_{i+1,i+1} -\chi,\\
& \text{ad}_\chi^{\mathfrak{J}} \, \bold{\bar{e}}_{i, i+1}(\mathcal{A}_{i, i}\mathcal{A}_{i+1, i+1}) = -(-1)^{i+\delta_i}\bar{\bold{e}}_{ii}-(-1)^{i+\delta_i} \bar{\bold{e}}_{i+1,i+1}+(-1)^{\delta_i}\chi.
\end{aligned}
\end{equation*}
Finally, one can check $ \text{ad}_\chi^{\mathfrak{J}} \, \bold{\bar{e}}_{i, i+1}(rdet{\mathcal{A}})=0$ so that $w_m \in \mathcal{W}(\overline{\mathfrak{gl}},f)$ for any $m\in I$.
Furthermore, one can find a basis $\{v_1, v_2\cdots,v_{4n+1}\}$ of $\mathfrak{gl}^{f}$ such that $\Delta_{\bar{v}_i}=\frac{i}{2}$. Hence, by checking the conditions of Proposition \ref{Prop:property_generator}, one can show that $w_1, w_2, \cdots, w_{4n+1}$ freely generate $\mathcal{W}(\overline{\mathfrak{gl}},f)$.
\end{proof}


\begin{theorem} \label{Theorem:osp(2n+1|2n)}
For the map $\kappa$ in \eqref{kappa},  if $w \in \mathcal{W}(\overline{\mathfrak{gl}}, f)$ then $\kappa(w) \in \mathcal{W}(\bar{\g}, f)$. Hence, for $w_1, \cdots, w_{4n+1}$ in  Lemma~\ref{Lemma:W(gl,f_osp)}, $\kappa(w_1), \cdots, \kappa(w_{4n+1})$ are elements in $\mathcal{W}(\bar{\g}, f)$.
\end{theorem}

\begin{proof}
Let $w\in \mathcal{W}(\overline{\mathfrak{gl}}, f)$. Consider the decomposition 
\[ w= \pi_1(w)+ \pi_2(w)\]
via \eqref{decomposition}. Take $n \in \n=\g_{>0}$. Denote by $\mathfrak{I}_{\mathfrak{gl}}$  the ideal  $\mathfrak{J}$ of $\mathcal{V}(\overline{\mathfrak{gl}})$ introduced in Proposition~\ref{prop:another definition} and by $\mathfrak{J}_{\mathfrak{osp}}$ the ideal $\kappa(\mathfrak{J}_{\mathfrak{gl}})$ of $\mathcal{V}(\bar{\g})$. 
Then 
\[ \{\bar{n} {}_\chi w \} \in \mathbb{C}[\chi]\otimes \mathfrak{I}_{\mathfrak{gl}}.\]
 Hence $\{\bar{n}{}_\chi \pi_1(w)\} \subset \mathbb{C}[\chi] \otimes(  \,  \mathfrak{I}_{\mathfrak{gl}} \,   \cap \,   S(R) \, )$. Since $\psi$ in \eqref{kappa} is a SUSY PVA homomorphism, 
\[ \{\psi(\bar{n}){}_\chi \kappa (w)\} \in   \mathbb{C}[\chi]\otimes \kappa(\mathfrak{I}_{\mathfrak{gl}}) =   \mathbb{C}[\chi]\otimes \mathfrak{I}_{\mathfrak{osp}}. \]
Therefore $\kappa(w) \in \mathcal{W}(\bar{\g}, f)$.
\end{proof}

Recall the basis $\{F_{ij}\}_{(i,j)\in \mathcal{B}}$ of $\g$ in Example \ref{osp_4n pm1}. Then for
 \[\mathcal{C}:= \{(i,j)\in \mathcal{B}| i\geq j\},\]
the subset $\{F_{ij}\}_{(i,j)\in \mathcal{C}}$ and $\{F_{ij}\}_{(j,i)\in \mathcal{C}}$  are bases of $\mathfrak{p}$ and $\mathfrak{p}_+$, respectively.  Consider   
\[\  \  E_{ji}:= \frac{(-1)^{p(j)}}{2}F_{ji} = \frac{(-1)^{j+1}}{2}F_{ji} \in \g \quad\text{ for }  i,j\in I\]
and the matrix $\mathcal{A}_{\mathfrak{osp} }^{(2n+1|2n)}$: 
\begin{equation*}
\begin{aligned}
 \tiny \left( \begin{array}{ccccccccc} D+\bar{E}_{11} & \bar{E}_{21}&  \cdots & \bar{E}_{n+1\, 1} &  \bar{E}_{n+2\, 1}& \cdots& \bar{E}_{4n\, 1}& \bar{E}_{4n+1\, 1}  \\
 -(-1)^{\delta_1}& D+\bar{E}_{22}  & \cdots &  \bar{E}_{n+1\, 2} &  \bar{E}_{n+2\, 2}& \cdots &\bar{E}_{4n\, 2}& \bar{E}_{4n+1\, 2}\\
 \vdots & \vdots&&\vdots&\vdots&&\vdots&\vdots \\
 0 & 0 & \cdots & D+ \bar{E}_{n+1\, n+1} & \bar{E}_{n+2\, n+1} &  \cdots &\bar{E}_{4n\, n+1}& \bar{E}_{4n+1\, n+1}\\
  0 & 0& \cdots &-(-1)^{\delta_{n+1}} & D+ \bar{E}_{n+2\, n+2} &  \cdots & \bar{E}_{4n\, n+2}& \bar{E}_{4n+1\, n+2}\\
   \vdots & \vdots&&\vdots&\vdots&&\vdots&\vdots \\
     0 & 0& \cdots &0 &0&  \cdots &D+\bar{E}_{4n\, 4n}&\bar{E}_{4n+1\, 4n}\\
          0 & 0& \cdots &0 &0&  \cdots & -(-1)^{\delta_{4n}} &D+\bar{E}_{4n+1\, 4n+1}\\
   \end{array}\right),
\end{aligned}
\end{equation*}
whose $ij$-entry $\mathcal{A}_{ij}= \left\{ \begin{array}{ll} \delta_{ij}D + \bar{E}_{ji} & \text{ if } j\geq i \\ -\delta_{i,j+1} (-1)^{\delta_j} & \text{ if } j<i\end{array} \right.$ is in $\mathcal{V}(\bar{\mathfrak{p}})[D]$.

\begin{corollary} \label{Cor:matrix_osp}
Let $\{F^{ij}\}_{(i,j)\in \mathcal{C}}$ be the basis of $\mathfrak{p}$ such that $(F^{ij}|F_{kl})= \delta_{jk} \delta_{il}$ and
let \[q_{\mathfrak{osp}}^{(2n+1|2n)} = \sum_{(i,j)\in \mathcal{C}}F_{ji}\otimes \bar{F}^{ij} \in \mathfrak{p}_+ \otimes \mathcal{V}( \bar{\mathfrak{p}}).\]  Then
\[ \mathcal{A}_{\mathfrak{osp}}^{(2n+1|2n)} = \sum_{i\in I}\bold{e}_{ii}D+ q_{\mathfrak{osp}}^{(2n+1|2n)}-f \otimes 1\in \mathfrak{gl}\otimes\mathcal{V}(\bar{\mathfrak{p}})[D].\]
Moreover, if we denote $rdet (\mathcal{A}_{\mathfrak{osp}}^{(2n+1|2n)}) = \sum_{i=0}^{4n+1} w_{4n+1-i} D^i$ for $w_i \in \mathcal{V}( \bar{\mathfrak{p}})$ then \[w_1, w_2, \cdots, w_{4n+1}\in \mathcal{W}(\bar{\g}, f).\] 

\end{corollary}

\begin{proof}
Note that 
\[
F^{ij}=  \left\{ \begin{array}{ll} \  E_{ij} & \text{ if } (i,j) \in \mathcal{C} \text{ and }  i+j \neq 4n+2, \\ \frac{1}{2} E_{ij}, & \text{ if } (i,j) \in \mathcal{C} \text{ and }  i+j =4n+2, \end{array}\right.
\]
and 
$ E_{ij} = \tau(j,i) E_{j'\, i'} $
for $\tau$ in \eqref{Theta} and $i'= 4n+2-i$.
Hence, for $(i,j) \in \mathcal{C}$ such that $i+j<4n+2$, we have
\begin{equation*}
 F_{ji} \otimes \bar{F}^{ij}=   F_{ji} \otimes \bar{E}_{ij}  =  \bold{e}_{ji} \otimes \bar{E}_{ij}+ \tau(j,i)\bold{e}_{i' \, j'}  \otimes \bar{E}_{ij} = \bold{e}_{ji} \otimes  \bar{E}_{ij}+\bold{e}_{i' \, j'} \otimes \bar{E}_{j'\, i'} .
\end{equation*}
If $i=j'$ and $p(i)=1$ (resp. $p(i)=0$), then 
$
 F_{ji} \otimes \bar{F}^{ij}=  2\bold{e}_{ji} \otimes \frac{1}{2}  \bar{E}_{ij}=   \bold{e}_{ji}\otimes \bar{E}_{ij}
$ 
(resp. $ F_{ji}\otimes \bar{F}^{ij} =0=\bold{e}_{ji}\otimes \bar{E}_{ij}$). Hence the first assertion holds.

In addition, since $\kappa(\bar{\bold{e}}^{ji})= \bar{E}^{ji}$, we have
\[\  \   q_{\mathfrak{osp}}^{(2n+1|2n)}=\sum_{\substack{i\leq j\\ i,j\in I}} \bold{e}_{ij}\otimes \kappa(\bar{\bold{e}}^{ji}) = \sum_{\substack{i\leq j\\ i,j\in I}} \bold{e}_{ij}\otimes \bar{E}^{ji}.\]
 By Lemma~\ref{Lemma:W(gl,f_osp)} and Theorem~\ref{Theorem:osp(2n+1|2n)}, we get the second assertion.
\end{proof}

Let us find a minimal generating set of $\mathcal{W}(\bar{\g}, f)$. Since 
$ \g^{f}:= \ker( \text{ad}\, f )$
has the same dimension as $\dim \g_0,$ the dimension of $\g^{f}$ is $2n$. In addition, since  
\[
\dim \g_k= \left\{ \begin{array}{ll} k+2n+\frac{1}{2} & \text{ if } 2k \text{ is odd }, \\
k+2n+1& \text{ if } 2k \equiv 2 \Mod 4,\\
k+2n &\text{ if } 2k \equiv 0 \Mod 4, 
  \end{array}\right.
 \]
for $-2n\leq k \leq 0$, there is a basis 
\begin{equation} \label{basis of ker ad f}
V^f:=\big\{\,  v_t \in \g_{\frac{1-t}{2}} \, | \, t\in I, t\equiv 0\text{ or }3 \Mod4\big\}
\end{equation} 
of $\g^{f}$.
Hence, by Proposition \ref{Prop:property_generator},  if we have that $\Delta_{w_t\in \mathcal{W}(\bar{\g},f)}=\frac{t}{2}$ and it has the form 
\begin{equation}\label{generator_form}
w_t =\bar{v}_t + (\text{polynomial degree $\geq 2$ terms}) +\text{ (total derivative part)},
\end{equation}
then $\{w_t | v_t \in V^f\}$ freely generates $\mathcal{W}(\bar{\g},f)$ as a differential algebra. In the following lemma, we find a basis $V^f$ of $\g^f$ satisfying \eqref{generator_form}.

\begin{lemma} \label{Lemma1}
\begin{enumerate}
\item Let $k\in\mathbb{Z}$ such that $1 \leq k \leq 4n+1$. Consider the element
\begin{equation}\label{sum of E's}
 v_{k}:=\sum_{i=1}^{4n+2-k} (-1)^{ki}\,E_{k-1+i\, i}\, (-1)^{\sum_{\ell=0}^{k-2}\delta_{i+\ell}}.
\end{equation}
Then 
$v_{k}=0$ if $k\equiv 1$ or 2 $\Mod 4$ and $v_{k}\neq 0$ if $k\equiv 0$ or 3 $\Mod 4$.
\item For any $k$ such that $1 \leq k \leq 4n+1$, the elements $v_{k}$ in \eqref{sum of E's}  are in $\g^{f}$. 
\end{enumerate}
\end{lemma}
\begin{proof}
(1) can be obtained by direct computations. For (2) 
let us look at the case where $k=3 \text{ or }4$. If $k=3$, 
\begin{equation} \label{k=3}
\begin{aligned}
& v_3=  \  \    \sum_{i=1}^{4n-1} (-1)^{i} E_{i+2, i} (-1)^{\delta_i+ \delta_{i+1}} 
=-\sum_{i=1}^{4n-1}  \bold{e}_{i+2, i} (-1)^{\delta_{i}+ \delta_{i+1}}.
\end{aligned}
\end{equation}
Now, since 
\begin{equation}
\begin{aligned}
\,  [f, \bold{e}_{i+2, i} (-1)^{\delta_{i}+ \delta_{i+1}}] =  \bold{e}_{i+3, i} (-1)^{\delta_i+ \delta_{i+1} + \delta_{i+2}} -\bold{e}_{i+2, i-1} (-1)^{\delta_{i-1}+ \delta_i+ \delta_{i+2}},& 
\end{aligned}
\end{equation}
we conclude $v_{k} \in \g^{f}$. 
By the same argument, one can check that $v_{k} \in  \g^{f}$ for $k=1, \cdots, 4n+1$. 
\end{proof} 

\begin{lemma}\label{Lemma2}
Recall the elements $w_1, w_2, \cdots, w_{4n+1}$ of  $\mathcal{W}(\bar{\g}, f)$ in Corollary \ref{Cor:matrix_osp}. Then we have $\Delta_{w_{4n+1-t}}=\frac{4n+1-t}{2}$.
\end{lemma}
\begin{proof}
Define a gradation $\Delta$ on $\mathcal{V}(\bar{\mathfrak{p}})[D]$:
\[ \Delta_{\bar{a}} = \frac{1}{2}-j_a, \quad \Delta_D= \frac{1}{2}, \quad \Delta_{AB}= \Delta_A + \Delta_B\]
where $a \in \g_{j_a}$ and $A,B\in \mathcal{V}(\bar{\mathfrak{p}})$. Note that for an element of $\mathcal{V}(\bar{\mathfrak{p}})$, this gradation coincides with the $\Delta$-grading in \eqref{Eqn:conformal weight}. Since $\Delta_{\mathcal{A}_{ij}}=\frac{j-i+1}{2}$ for $i\leq j-1$, we have $\Delta_{rdet(\mathcal{A})} = \frac{4n+1}{2}$. Hence $\Delta_{w_{4n+1-t}}= \frac{4n+1-t}{2}$.
\end{proof}

\begin{theorem}
Let $rdet(\mathcal{A}_{\mathfrak{osp}}^{(2n+1|2n)})= \sum_{t=0}^{4n+1} w_{4n+1-t} D^t$. Then \[\{ \, w_t\, | \, t\in I, \ t \equiv 0, 3  \ (\text{mod $4$}) \}\] generates $\mathcal{W}(\bar{\g}, f)$ as a differential algebra. In other words, 
\[ \mathcal{W}(\bar{\g}, f)= \mathbb{C}[w_t^{(m)}|m\in \mathbb{Z}_{\geq 0}, \, \, t\in I, \ t \equiv 0, 3 \Mod4 ]. \]
\end{theorem}

\begin{proof}
Note that if there are elements $w_t\in \mathcal{W}(\bar{\g},f)$ which have the form \eqref{generator_form} for $v_t$ in Lemma \ref{Lemma1} then they form a generating set of $\mathcal{W}(\bar{\g},f) $.
Let us denote by $w^1_t$ the linear terms in $w_t$ which are not total derivatives.  Thanks to Lemma \ref{Lemma2}, it is enough to show that 
\[w^1_t= c_t \bar{v}_t,\]
where $c_t$'s are  nonzero constants. Note that $w^1_t$ is the linear term without derivatives in  
\[ \sum_{i=1}^{4n+1-t}\mathcal{A}_{11}\cdots \mathcal{A}_{i-1,i-1} \mathcal{A}_{i, i+t}  \mathcal{A}_{t+i+1, t+i+1} \cdots \mathcal{A}_{4n+1, 4n+1}(-1)^{n-\sum_{j=i}^{t+i-1}\delta_j},\]
where $\mathcal{A}_{ij}$ is the $ij$-entry of $\mathcal{A}_{\mathfrak{osp}}^{(2n+1|2n)}$.
One can easily check that $ w_t^1 = (-1)^{t} \bar{v}_t.$

Now, if we denote $\mathcal{I}:=\{\, t\in I\, |   \, t\in I, \ t \equiv 0, 3  \ (\text{mod $4$}) \}$ then   
since $\{v_t| t\in \mathcal{I} \}$ is a basis of $\g^{f}$ we can conclude that $\mathcal{W}(\bar{\g},f)= \mathbb{C}[w_t^{(n)}|n\in \mathbb{Z}_{\geq 0}, \, t\in \mathcal{I}]$. 
\end{proof}


\subsection{$\mathcal{W}(\overline{\mathfrak{osp}}(2n-1|2n),f)$} \label{Sec: osp(2n-1|2n)} \

As in Section \ref{Sec:osp(2n+1|2n)}, consider the elements
\[\  \  E_{ji}:=  \frac{(-1)^j}{2}F_{ji}\in \g \quad\text{ for }  i,j\in I, \]
and the $(4n-1)\times (4n-1)$ matrix $\mathcal{A}_{\mathfrak{osp} }^{(2n-1|2n)}$ 
whose $ij$-entry \[\mathcal{A}_{ij}= \left\{ \begin{array}{ll} \delta_{ij}D + \bar{E}_{ji} & \text{ if } j\geq i \\ -\delta_{i,j+1} (-1)^{\delta_j} & \text{ if } j<i\end{array} \right.\] is in $\mathcal{V}(\bar{\mathfrak{p}})[D]$.
Using the same arguments as in Section \ref{Sec:osp(2n+1|2n)}, we get the following theorem.

\begin{theorem} \label{thm:osp(2n-1|2n)}
Let $\mathcal{C}:=\{ (i,j)\in\mathcal{B}| i\geq j\}$ and
$\{F^{ji}\}_{(i,j)\in \mathcal{C}}$ be the basis of $\mathfrak{p}_+$ such that $\left(F^{ij}|F_{kl}\right)=\delta_{jk}\delta_{il}$. Then
for $q_{\mathfrak{osp}}^{(2n-1|2n)} = \sum_{(i,j)\in \mathcal{C}}F_{ij}\otimes \bar{F}^{ji} \in \mathfrak{p}_+ \otimes \mathcal{V}( \bar{\mathfrak{p}})$, we have 
\[ \mathcal{A}_{\mathfrak{osp}}^{(2n-1|2n)} = \sum_{i\in I}\bold{e}_{ii}D+ q_{\mathfrak{osp}}^{(2n-1|2n)}-f \otimes 1\in \mathfrak{gl}\otimes\mathcal{V}(\bar{\mathfrak{p}})[D].\]
Moreover, if we denote $rdet (\mathcal{A}_{\mathfrak{osp}}^{(2n-1|2n)}) = \sum_{i=0}^{4n-1} w_{4n-1-i} D^i$ for $w_i \in \mathcal{V}( \bar{\mathfrak{p}})$ then \[w_1, w_2, \cdots, w_{4n-1}\in \mathcal{W}(\bar{\g}, f_{\mathfrak{osp}}).\] 

\end{theorem}\qed
\\
Furthermore, by observing the linear part of $w_t$ for each $t$, one can find a minimal generating set of $\mathcal{W}(\bar{\g},f).$

\begin{theorem}
  Let $rdet\big(\mathcal{A}_{\mathfrak{osp}}^{(2n-1|2n)}\big)= \sum_{t=0}^{4n-1} w_{4n-1-t} D^t$ as in Theorem \ref{thm:osp(2n-1|2n)}. Then \[\{ \, w_t\, | \, t\in I, t\equiv 0 \text{ or } 3 \Mod4\, \}\] generates $\mathcal{W}(\bar{\g}, f)$ as a differential algebra. More precisely,
  \[ \mathcal{W}(\bar{\g}, f)= \mathbb{C}[w_t^{(m)}|m\in \mathbb{Z}_{\geq 0}, \,  t\in I, t\equiv 0 \text{ or } 3 \Mod4\, ]. \]
  \end{theorem}

  \begin{proof}
Define a gradation $\Delta$ on $\mathcal{V}(\bar{\mathfrak{p}})[D]$ by 
\[ \Delta_{\bar{a}} = \frac{1}{2}-j_a, \quad \Delta_D= \frac{1}{2}, \quad \Delta_{AB}= \Delta_A + \Delta_B\]
where $a \in \g_{j_a}$ and $A,B\in \mathcal{V}(\bar{\mathfrak{p}})$. Then $\Delta_{rdet(\mathcal{A})} = \frac{4n-1}{2}$, so $\Delta_{w_{t}}= \frac{t}{2}$. One can find a basis
\begin{equation}
  V^f:=\big\{\, v_t\in \g_{\frac{1-t}{2}}\ | \ t\in I, t\equiv 0 \text{ or }3 \Mod4 \, \big\}
\end{equation}
of $\g^f$ such that 
\begin{equation}
  w_t=\bar{v}_t+\left(\text{polynomial degree}\geq2 \text{ terms}\right)+\left(\text{total derivative part}\right).
\end{equation}
Since $\Delta_{\bar{v}_t}=\frac{t}{2}=\Delta_{w_t}$, we get the conclusion by Proposition \ref{Prop:property_generator}.
  \end{proof}

\

\section{Generators of $\mathcal{W}(\mathfrak{osp}(M|2n), f)$ for $M=2n,2n+2$ \label{sec:ospeven}}

In this section, we find generators of  $\mathcal{W}(\mathfrak{osp}(M|2n), f)$ for $M=2n,2n+2$ 
using the matrix presentation of $\mathfrak{osp}(M|2n)$,
the index set $I=I_{\bar{0}} \sqcup I_{\bar{1}}$,  the elements $F_{ij}$ for $i,j\in I$ and the odd principal nilpotent $f$ introduced in Example \ref{osp_4n +0,2}.

\subsection{$\mathcal{W}(\overline{\mathfrak{osp}}(2n|2n),f)$} \label{Sec:W(osp(2n|2n))}\

Let $\g= \mathfrak{osp}(2n|2n)$.
We introduce the $(4n+1)\times(4n+1)$ matrix $\mathcal{A}_{\mathfrak{osp}}^{(2n|2n)}$
 whose entries are in $\mathcal{V}(\bar{\mathfrak{p}})(\!(D^{-1})\!)$, obtained by the following procedure:

 \begin{itemize}
 \item[(Step1)] Consider the $4n\times4n$ matrix 
\[\  \  \mathcal{M}=\sum_{i=1}^{4n} \bold{e}_{ii}\otimes D +q_{\mathfrak{osp}} - f\otimes 1 \in \mathfrak{gl}(2n|2n)\otimes \mathcal{V}(\bar{\mathfrak{p}})[D]\]
where  $q_{\mathfrak{osp}}= \sum_{i\geq j\in I} \bold{e}_{ji}\otimes \bar{E}_{ij}$ for 
$E_{ij}:=\frac{1}{2}(-1)^{p(i)}F_{ij}.$
 \item[(Step2)]  Consider the $4n\times 4n$ matrix $\mathcal{N}$  obtained from $\mathcal{M}$
 by replacing the $2n$-th column by [$(2n)$-th column of $\mathcal{M}$]-[$(2n+1)$-th column of $\mathcal{M}$].
 \item[(Step3)] Consider the $4n\times 4n$ matrix $\widetilde{\mathcal{A}}$  obtained from $\mathcal{N}$
 by replacing the $(2n+1)$-th row by [$(2n+1)$-th row of $\mathcal{M}$]-[$(2n)$-th row of $\mathcal{M}$].
 \item[(Step4)] Let $\mu_{11}$, $\mu_{12}$, $\mu_{21}$ and $\mu_{22}$ be $2n\times 2n$-matrices such that 
 \[ \widetilde{\mathcal{A}}= \left( \begin{array}{cc} \mu_{11} & \mu_{12} \\ \mu_{21} & \mu_{22}\end{array} \right).\]
  \item[(Step5)] The $(4n+1)\times(4n+1)$ matrix $\mathcal{A}_{\mathfrak{osp}}^{(2n|2n)}$ has $\mu_{11}$, $\mu_{12}$, $\mu_{21}$ and $\mu_{22}$ as the upper-left, upper-right, lower-left and lower-right $2n\times 2n$-matrices, respectively. The only nonzero entry in the $(2n+1)$-th row and the $(2n+1)$-th column of $\mathcal{A}_{\mathfrak{osp}}^{(2n|2n)}$ is the $(2n+1, 2n+1)$-entry, which is $D^{-1}$.
 \end{itemize}
 Hence we get 
 \begin{equation}  \mathcal{A}_{\mathfrak{osp}}^{(2n|2n)}=\label{matix(2n|2n)}
\small{
\begin{pmatrix}
\begin{matrix}
&&\\
& \mu_{11}&\\
&&\\
\end{matrix}
  & \rvline & 0& \rvline&\begin{matrix}
&&\\
& \mu_{12}&\\
&&\\
\end{matrix} \\
\hline
0 & \rvline &
D^{-1}&
  \rvline & 0\\
\hline
\begin{matrix}
&&\\
& \mu_{21}&\\
&&\\
\end{matrix}
& \rvline & 0
&
  \rvline & 
\begin{matrix}
&&\\
& \mu_{22}&\\
&&\\
\end{matrix}
\end{pmatrix},
}
\end{equation}
where the entries of the matrices $\mu_{11}$, $\mu_{22}$, $\mu_{12}$ and $\mu_{21}$ are given by 
\begin{eqnarray} \label{mu11}
&&(\mu_{11})_{ij}= \left\{\begin{array}{ll} \delta_{2n\, i} D + \bar{E}_{2n\, i}-\bar{E}_{2n+1\, i} & \text{ if } i\leq j=2n, \\ \delta_{ij}D+ \bar{E}_{ji} & \text{ if } i\leq j <2n, \\ -\delta_{i-1\, j}; & \text{ otherwise},
\end{array} \right.
\\
  \label{mu22}
&&(\mu_{22})_{2n+1-i\ 2n+1-j} = \left\{\begin{array}{ll} \delta_{2n\, j} D + \bar{E}_{j'\, (2n)'}-\bar{E}_{j'\, (2n+1)'} & \text{ if } j\leq i=2n, \\ \delta_{ij}D+ \bar{E}_{j'i'} & \text{ if } j\leq i <2n, \\ (-1)^j\delta_{i\, j-1} & \text{ otherwise},
\end{array} \right.
\\
&&(\mu_{12})_{i\ 2n+1-j}= \bar{E}_{j'\, i}\quad \text{and} \quad (\mu_{21})_{i\, j}= -2\delta_{i\, 1} \delta_{j\, 2n} D
\quad\text{for } i,j\in \{1,2,\cdots, 2n\}
\end{eqnarray}
where we used the notation $i'= 4n+1-i$.

Denote the row determinant of the upper-left $k \times k$ submatrix of $\mu_{11}$ in \eqref{mu11} by $A_k$ and 
the one of the lower-right $k \times k$ submatrix of $\mu_{22}$ in \eqref{mu22} by $B_k$. Letting $A_0=B_0=1$, the row determinant of $\mathcal{A}_{\mathfrak{osp}}^{(2n|2n)}$ is
\begin{equation} \label{cdet osp(2n|2n)}
  rdet(\mathcal{A}_{\mathfrak{osp}}^{(2n|2n)})=A_{2n}D^{-1}B_{2n}+2\sum_{j,k=1}^{2n}{(-1)^{k+T_{2n}(k)}A_{j-1}\bar{E}_{k'j}B_{k-1}},
\end{equation}
where 
\begin{equation}\label{T, sign}
T_i(k):=\delta_i+\delta_{i-1}+\cdots+\delta_{k+1}=\floor*{\frac{i+1}{2}}-\floor*{\frac{k+1}{2}} \quad\text{ if } i>k
\end{equation}
 and $T_i(k):=0$ otherwise.

\begin{lemma} \label{Lemma A_k, B_k}
The following identities hold 
\begin{eqnarray} \label{cdet A_k}
&&\begin{aligned}
& A_k=A_{k-1}(D+\bar{E}_{kk})+\sum_{i=1}^{k-1}{A_{i-1}\bar{E}_{ki}} \quad \text{ if } k=1, \cdots, 2n-1, \\
& A_{2n}=A_{2n-1}(D+\bar{E}_{2n\, 2n})+\sum_{i=1}^{2n-1}{A_{i-1}(\bar{E}_{2n\, i}-\bar{E}_{2n+1\, i})}.
  \end{aligned}
\\ 
\label{cdet B_k}
&&  \begin{aligned}
& B_k=(D+\bar{E}_{k'k'})B_{k-1}+\sum_{i=1}^{k-1}{(-1)^{i+k+T_k(i)}\bar{E}_{i'k'}B_{i-1}}\quad 
\text{ if } k=1, \cdots, 2n-1, \\
& B_{2n}=(D+\bar{E}_{(2n)'\, (2n)'})B_{2n-1}+\sum_{i=1}^{2n-1}{(-1)^{i+T_{2n}(i)}(\bar{E}_{i'\, (2n)'}-\bar{E}_{i'\, (2n+1)'})B_{i-1}}.
  \end{aligned}
\end{eqnarray}

\end{lemma} 

In order to see the relation between $A_k$ and $B_k$, let us introduce the adjoint operator ${}^*: \mathcal{V}(\bar{\mathfrak{p}})[D] \to \mathcal{V}(\bar{\mathfrak{p}})[D] $ defined by  
\begin{equation}
\begin{aligned}
& a^*:=a, \\
& (a D^m)^*:= (-1)^{m\cdot p(a)+\lfloor \frac{m+1}{2}\rfloor}D^m a, \\
& (f(D) g(D))^*= (-1)^{p(f)p(g)} g(D)^* f(D)^*,
\end{aligned}
\end{equation}
for $m\in \mathbb{Z}_{\geq 0}$, $a\in \mathcal{V}(\bar{\mathfrak{p}})$ and $f(D), g(D)\in \mathcal{V}(\bar{\mathfrak{p}})[D]$. Note that $f(D)^{**}=f(D)$.

\begin{lemma} \label{Lem:const term of A_2n}
For $k\leq 2n$, we have
\[ (B_k)^*= (-1)^{\lfloor \frac{k+1}{2}\rfloor}A_k.\]
Equivalently, $ (A_k)^*= (-1)^{\lfloor \frac{k+1}{2}\rfloor}B_k$.
\end{lemma}
\begin{proof}
If $k=1$ then
\[ (B_1)^*= (D+\bar{E}_{1'1'})^*=-D-\bar{E}_{11}=-A_1.\]
Inductively, for any $k< 2n$  
\begin{equation}
\begin{aligned}
(B_k)^* & =  \sum_{i=1}^{k} (-1)^{ i+k+T_k(i)  }(-1)^{(k+1)(i+1)}(B_{i-1})^*(-\delta_{ik}D+\bar{E}_{i'k'})\\
& =\sum_{i=1}^{k} (-1)^{\lfloor \frac{k+1}{2}\rfloor} A_{i-1}(\delta_{ik}D +\bar{E}_{k\, i})  =(-1)^{\lfloor \frac{k+1}{2}\rfloor} A_k.
\end{aligned}
\end{equation}
When $k=2n$, a similar proof works.
\end{proof}

\begin{corollary}\label{cor:(2n|2n), A and B}
Let us denote 
\[
  A_{2n}=D^{2n}+a_{1}D^{2n-1}+a_{2}D^{2n-2}+\cdots+a_{2n-1}D+a_{2n},\]
for $a_i\in\mathcal{V}(\bar{\frak{p}}).$ Then for $b_{2n-i}=(-1)^{n+i+\lfloor\frac{i+1}{2} \rfloor}a_{2n-i}$, we have 
\[  B_{2n}=D^{2n}+D^{2n-1}b_{1}+D^{2n-2}b_{2}+\cdots+Db_{2n-1}+b_{2n}.\]
In particular, $b_{2n}=(-1)^{n+1}a_n$.
\end{corollary}

\begin{proof}
Since $(a_{2n-i}D^{i})^*=(-1)^i(-1)^{\lfloor \frac{i+1}{2}\rfloor}D^i a_{2n-i}$ and $(A_{2n})^*= (-1)^{n} B_{2n}$, the corollary follows.
\end{proof}

\begin{lemma} \label{Lem:bracket for submatrix}
  For $A_k$ and $B_k$ as in Lemma \ref{Lemma A_k, B_k} and $1\leq i \leq 2n-1$, 
  the following identities hold.
\begin{eqnarray}
\label{item1a} 
&&\text{ad}_\chi^{\mathfrak{J}}\,  \bar{F}_{i\, i+1} (A_k )=
\left\{\begin{array}{ll} (-1)^iA_{i-1}(D+\chi) & \text{ if } k=i, \\
  0 & \text{ otherwise. } \end{array} \right.
\\
\label{item1b} 
&&\text{ad}_\chi^{\mathfrak{J}} \, \bar{F}_{2n-1\, 2n+1} ( A_k) =
\left\{\begin{array}{ll} -A_{2n-2}(D+\chi) & \text{ if } k=2n-1, \\
  -2A_{2n-2}(D+\chi)D & \text{ if } k=2n, \\
  0 & \text{ otherwise. } \end{array} \right.
\\
\label{item2a} 
&&\text{ad}_\chi^{\mathfrak{J}}\, \bar{F}_{i\, i+1} ( B_k )=
\left\{\begin{array}{ll} B_{i-1} & \text{ if } k=i, \\
  0 & \text{ otherwise. } \end{array} \right.
\\
\label{item2b} 
&&\text{ad}_\chi^{\mathfrak{J}} \,  \bar{F}_{2n-1\, 2n+1} ( B_k) =
\left\{\begin{array}{ll} 
  B_{2n-2} & \text{ if } k=2n-1, \\
  -2(D+\chi)B_{2n-2} & \text{ if } k=2n, \\
  0 & \text{ otherwise. } \end{array} \right.
\end{eqnarray}
Here, $A_k(D+\chi)$ denotes $A_k$ with the symbol $D$ replaced by $D+\chi$.

\end{lemma}

\begin{proof}
Since all relations  can be proved in a similar way, let us only show 
 \eqref{item2b}.
 
 It is clear that  $\text{ad}_\chi^{\mathfrak{J}} \,  \bar{F}_{2n-1\, 2n+1} ( B_k )=0$ when $k\leq 2n-2$. When $k=2n-1$, we have 
 \begin{align*}
\text{ad}_\chi^{\mathfrak{J}}\,  \bar{F}_{2n-1,2n+1} ( B_{2n-1})&=  \text{ad}_\chi^{\mathfrak{J}}\, \bar{F}_{2n-1,2n+1} ( \bar{E}_{(2n-1)' \, (2n-1)'})B_{2n-2} \\
& \  \  \quad +\sum_{k=1}^{2n-2}{(-1)^{1+k+T_{2n-1}(k)}\text{ad}_\chi^{\mathfrak{J}}\bar{F}_{2n-1,2n+1} ( \bar{E}_{k'(2n-1)'})B_{k-1}} \\
&=-\bar{E}_{2n-1\,2n+1}B_{2n-2} =B_{2n-2}.
\end{align*}
Here, the first equality holds by \eqref{cdet B_k}.
Furthermore, 
\begin{equation}\label{F_{2n-1, 2n+1}, B_2n}
\begin{aligned}
&\text{ad}_\chi^{\mathfrak{J}}\,  \bar{F}_{2n-1,2n+1} ( B_{2n})=  \text{ad}_\chi^{\mathfrak{J}}\, \bar{F}_{2n-1,2n+1} ( \bar{E}_{(2n)' \, (2n)'}B_{2n-1}) \\
& \  \  \qquad +\sum_{k=1}^{2n-1}{(-1)^{k+T_{2n}(k)}\text{ad}_\chi^{\mathfrak{J}}\, \bar{F}_{2n-1,2n+1} \big( (\bar{E}_{k'(2n)'}-\bar{E}_{k'(2n+1)'})B_{k-1}\big)} \\
&=-B_{2n-1}-(D+\chi+\bar{E}_{(2n)'(2n)'})B_{2n-2} \\
&  \  \  \qquad +\sum_{k=1}^{2n-1}{(-1)^{k+T_{2n}(k)}\text{ad}_\chi^{\mathfrak{J}}\, \bar{F}_{2n-1,2n+1} \big( (\bar{E}_{k'(2n)'}-\bar{E}_{k'(2n+1)'})B_{k-1}\big)} .
\end{aligned}
\end{equation}

Let us observe the last term in  \eqref{F_{2n-1, 2n+1}, B_2n}.
If  $k\neq 2n-1$ then 
\[
  \text{ad}_\chi^{\mathfrak{J}}\bar{F}_{2n-1,2n+1} \big( (\bar{E}_{k'(2n)'}-\bar{E}_{k'2n})B_{k-1}\big)=-\bar{E}_{k'(2n-1)'}B_{k-1}
\]
and if  $k=2n-1$ then 
  \begin{equation*}
  \text{ad}_\chi^{\mathfrak{J}}\bar{F}_{2n-1,2n+1}  \big((\bar{E}_{(2n-1)'\,(2n)'}-\bar{E}_{(2n-1)'\,2n})B_{2n-2}\big) =(-\bar{E}_{(2n-1)'\,(2n-1)'}-\bar{E}_{(2n)'(2n)'}+\chi)B_{2n-2}.
  \end{equation*}
Therefore, by \eqref{F_{2n-1, 2n+1}, B_2n}, we have 
\begin{align*}
&\text{ad}_\chi^{\mathfrak{J}}\{\bar{F}_{2n-1\,2n+1} \chi B_{2n}\}=-B_{2n-1}-(D+\chi+\bar{E}_{(2n)'(2n)'})B_{2n-2}\\
&  \  \  \ \  +\sum_{k=1}^{2n-2}(-1)^{k+T_{2n}(k)+1}\bar{E}_{k'(2n-1)'}B_{k-1}+(\bar{E}_{(2n-1)'\,(2n-1)'}+\bar{E}_{(2n)'(2n)'}-\chi)B_{2n-2}\\
&\ \ =-2(D+\chi)B_{2n-2}. \end{align*}
\end{proof}

Recall the matrix $\mathcal{A}_{\mathfrak{osp}}^{(2n|2n)}$ in \eqref{matix(2n|2n)} and $a_{2n}$ in Corollary~\ref{cor:(2n|2n), A and B}.  Then, by \eqref{cdet osp(2n|2n)} and Corollary \ref{cor:(2n|2n), A and B}, the row determinant of $\mathcal{A}_{\mathfrak{osp}}^{(2n|2n)}$ is
\begin{equation}\label{2n|2n, row determinant-2}
  rdet(\mathcal{A}_{\mathfrak{osp}}^{(2n|2n)})=D^{4n-1}+ \sum_{k=0}^{4n-2}w_{4n-1-k} D^k +(-1)^{n}a_{2n}D^{-1}a_{2n},
\end{equation}
where $w_k\in\mathcal{V}(\bar{\mathfrak{p}})$ for $k=1, \cdots, 4n-1$.

\begin{theorem} \label{thm:osp 4n elements}
Let $w_1, w_2, \cdots, w_{4n-1}$ be in \eqref{2n|2n, row determinant-2} and 
let $\widetilde{w}_{2n}:=a_{2n}$.
  Then  \[w_{1}, w_{2}, \cdots,w_{4n-1}, \widetilde{w}_{2n} \in \mathcal{W}(\bar{\g},f).\]
\end{theorem}
\begin{proof}
  Since  $\{F_{i\,i+1}\,|\,i=1,2,\cdots,{2n-1}\}\cup \{F_{2n-1\,2n+1}\}$ generates $\n$, the element $w_k$ is in $\mathcal{W}(\bar{\g},f)$ if and only if 
\begin{equation}\label{goal of the theorm}
\text{ad}_\chi^{\mathfrak{J}}\bar{F}_{i\,i+1}( w_k)=0 \text{ and }\text{ad}_\chi^{\mathfrak{J}}\bar{F}_{2n-1\,2n+1}( w_k )=0
\end{equation}
for $1\leq i \leq 2n-1$. Note that in Lemma \ref{Lem:bracket for submatrix},
 we already showed \eqref{goal of the theorm} for $\widetilde{w}_{2n}$.
Thus, $\widetilde{w}_{2n} \in \mathcal{W}(\bar{\g},f)$.

To see \eqref{goal of the theorm} for $k=1,2,\cdots, 4n-1$, we claim that 
\begin{equation}\label{Eq:aim_2} 
\text{ad}_\chi^{\mathfrak{J}}\bar{F}_{i\,i+1} ( rdet(\mathcal{A}_{\mathfrak{osp}}^{(2n|2n)}))=0 \text{ \ and\  }\text{ ad}_\chi^{\mathfrak{J}}\bar{F}_{2n-1\,2n+1}(  rdet(\mathcal{A}_{\mathfrak{osp}}^{(2n|2n)}))=0
\end{equation}
 for any $1\leq i \leq 2n-1$.    

By Lemma \ref{Lem:bracket for submatrix},
\begin{equation} \label{eq:(ii+1)}
    \begin{aligned}
   &   \text{ad}_\chi^{\mathfrak{J}}\bar{F}_{i\,i+1} \big( rdet(\mathcal{A}_{\mathfrak{osp}}^{(2n|2n)})\big) =2\sum_{k=1}^{2n}(-1)^{n+k+i+\floor*{\frac{k+1}{2}}}A_{i-1}(D+\chi)\bar{E}_{k'i+1}B_{k-1}\\
 &  \qquad  \quad  \hskip 2cm +2\sum_{j=1}^{2n}(-1)^{n+1+\floor*{\frac{i}{2}}}A_{j-1}(D+\chi)\bar{E}_{i+1'j}B_{i-1} \\
   &  \qquad  \quad  \hskip 2cm+2\sum_{j,k=1}^{2n}(-1)^{n+j+k+1+\floor*{\frac{k+1}{2}}}A_{j-1}(D+\chi)\text{ad}_\chi^{\mathfrak{J}}\bar{F}_{i\, i+1} ( \bar{E}_{k'j})B_{k-1}.
\end{aligned}
\end{equation}
Observe that the last term in \eqref{eq:(ii+1)} is
\begin{equation} \label{eq:(ii+1)-2}
\begin{aligned}
& \  \  2\sum_{j,k=1}^{2n}(-1)^{n+j+k+1+\floor*{\frac{k+1}{2}}}A_{j-1}(D+\chi)\text{ad}_\chi^{\mathfrak{J}}\bar{F}_{i\, i+1} ( \bar{E}_{k'j})B_{k-1}\\
&  =  \  \  2\sum_{k\neq i}(-1)^{n+i+k+1+\floor*{\frac{k+1}{2}}}A_{i-1}(D+\chi)\bar{E}_{k'i+1}B_{k-1} \\
& + \  \   2\sum_{j\neq i}(-1)^{n+i+\floor*{\frac{i+1}{2}}}A_{j-1}(D+\chi)\bar{E}_{i+1'j}B_{i-1}\\
& + \  \   2(-1)^{n+\floor*{\frac{i+1}{2}}}\left((-1)^{i+1}+1-1+(-1)^i\right)A_{i-1}(D+\chi)\bar{E}_{i+1'i}B_{i-1}.
\end{aligned}
\end{equation} 
By \eqref{eq:(ii+1)} and \eqref{eq:(ii+1)-2}, we get the first assertion in \eqref{Eq:aim_2}:
\begin{equation}
 \text{ad}_\chi^{\mathfrak{J}}\bar{F}_{i\,i+1} \big( rdet(\mathcal{A}_{\mathfrak{osp}}^{(2n|2n)})\big)=0.
\end{equation}

 
Let us show the second assertion in \eqref{Eq:aim_2}.
By Lemma \ref{Lem:bracket for submatrix} and the sesquilinearity, we have 
\begin{equation}\label{eq:(2n-1, 2n+1)}
    \begin{aligned}
   &\text{ad}_\chi^{\mathfrak{J}}\bar{F}_{2n-1\,2n+1} \big(  rdet(\mathcal{A}_{\mathfrak{osp}}^{(2n|2n)})\big) 
    =2A_{2n}(D+\chi)B_{2n-2}-2A_{2n-2}(D+\chi)B_{2n} \\
    & \qquad\qquad\qquad\qquad\qquad +2\sum_{j,k=1}^{2n}(-1)^{n+k+\floor*{\frac{k+1}{2}}}\text{ad}_\chi^{\mathfrak{J}}\bar{F}_{2n-1\,2n+1}\big( A_{j-1}\bar{E}_{k'j}B_{k-1}\big).
    \end{aligned}
\end{equation}

Take a look at the last line in \eqref{eq:(2n-1, 2n+1)}. By Lemma \ref{Lem:bracket for submatrix},
\begin{equation} \label{eq:(2n-1, 2n+1)-2}
  \begin{aligned}
  & \  \  \sum_{j,k=1}^{2n}(-1)^{n+k+\floor*{\frac{k+1}{2}}}\text{ad}_\chi^{\mathfrak{J}}\bar{F}_{2n-1\,2n+1}\big( A_{j-1}\bar{E}_{k'j}B_{k-1}\big) \\
  & \  \  =\sum_{k=1}^{2n}(-1)^{n+k+1+\floor*{\frac{k+1}{2}}}A_{2n-2}(D+\chi)\bar{E}_{k'\,2n}B_{k-1}\\
    & \  \  +\sum_{k=1}^{2n}(-1)^{n+j+k+1+\floor*{\frac{k+1}{2}}}A_{j-1}(D+\chi)\text{ad}_\chi^{\mathfrak{J}}\bar{F}_{2n-1\,2n+1}(\bar{E}_{k'\,j})B_{k-1}\\
  & \  \  +\sum_{k=1}^{2n}A_{j-1}(D+\chi)\bar{E}_{(2n)'\,j}B_{2n-2}.\\
\end{aligned}
\end{equation}
The second term in the RHS of \eqref{eq:(2n-1, 2n+1)-2} equals to
\begin{equation} \label{eq:(2n-1, 2n+1)-3}
  \begin{aligned}
  & \  \ \sum_{k=1}^{2n}(-1)^{n+j+k+1+\floor*{\frac{k+1}{2}}}A_{j-1}(D+\chi)\text{ad}_\chi^{\mathfrak{J}}\bar{E}_{k'\,j}B_{k-1} \\
  &= \  \ \sum_{k=1}^{2n-2}(-1)^{n+k+\floor*{\frac{k+1}{2}}}A_{2n-2}(D+\chi)\bar{E}_{k'\,(2n)'}B_{k-1}\\
  &+ \  \ \sum_{k=1}^{2n-2}(-1)^{n+k+1+\floor*{\frac{k+1}{2}}}A_{2n-1}(D+\chi)\bar{E}_{k'\,(2n-1)'}B_{k-1}\\
  &- \  \ \sum_{j=1}^{2n-2}A_{j-1}(D+\chi)\bar{E}_{2n\,j}B_{2n-2}-2A_{2n-2}(D+\chi)\bar{E}_{2n\,2n-1}B_{2n-2}\\
  &-A_{2n-1}(D+\chi)\left(\bar{E}_{2n-1\,2n-1}+\bar{E}_{2n\,2n}+\chi\right)B_{2n-2}\\
  & -\ \ \sum_{j=1}^{2n-2}A_{j-1}(D+\chi)\bar{E}_{2n-1\,j}B_{2n-1}\\
  &-A_{2n-2}(D+\chi)\left(\bar{E}_{2n-1\,2n-1}+\bar{E}_{2n\,2n}-\chi\right)B_{2n-1}.
  \end{aligned}
\end{equation}
Combining \eqref{eq:(2n-1, 2n+1)-2} with \eqref{eq:(2n-1, 2n+1)-3} and using Lemma \ref{Lemma A_k, B_k}, one can show that 
\begin{equation}\label{eq:(2n-1, 2n+1)-4}
\begin{aligned}
 \  \ \sum_{j,k=1}^{2n} & (-1)^{n+k+\floor*{\frac{k+1}{2}}}\text{ad}_\chi^{\mathfrak{J}}\bar{F}_{2n-1\,2n+1}\big( A_{j-1}\bar{E}_{k'j}B_{k-1}\big)\\
 & = A_{2n-2}(D+\chi)B_{2n}-A_{2n}(D+\chi)B_{2n-2}.
 \end{aligned}
 \end{equation}
By \eqref{eq:(2n-1, 2n+1)} and \eqref{eq:(2n-1, 2n+1)-4}, we get the second assertion of \eqref{Eq:aim_2}.
\end{proof}

  Now, we want to find a freely generating set of $\mathcal{W}(\bar{\g},f)$ using Proposition \ref{Prop:property_generator}. Since $f\in\g$ is a principal element, $\dim\g^f$ equals to the dimension of $\g_0$, that is, $\dim\g^f=2n$. Furthermore, we have that for $-n<k\leq 0$,
  \begin{equation}
    \dim\g_k=\left\{\begin{array}{lll} k+2n+\frac{1}{2} & \text{ if } 2k\text{ is odd, } \\
        k+2n & \text{ if } 2k\equiv 0 \Mod4, \\
        k+2n+1 & \text{ if } 2k\equiv 2 \Mod4. \end{array}\right.
  \end{equation}
Also for $-2n+1\leq k \leq -n$,
\begin{equation}
    \dim\g_k=\left\{\begin{array}{lll} k+2n-\frac{1}{2} & \text{ if } 2k\text{ is odd, } \\
        k+2n-1 & \text{ if } 2k\equiv 0 \Mod4, \\
        k+2n & \text{ if } 2k\equiv 2 \Mod4. \end{array}\right.
  \end{equation}
Hence, one can find a basis
\begin{equation}
    V^f:=\big\{v_l\in \g_{\frac{1-l}{2}}\,|\, 1 < l \leq 4n-1,\ l\equiv 0 \text{ or }3 \Mod 4\big\} \cup \big\{\widetilde{v}_{2n}\in \g_{\frac{1-2n}{2}}\big\}
  \end{equation}
of $\g^f$, consisting of homogeneous elements. Explicitly, $v_l$ and $\widetilde{v}_{2n}$ can be written as follows.

  \begin{lemma} \label{Lem:osp basis of keradf}
    \begin{enumerate}
      \item For an integer $l$ such that $2n+1 \leq l \leq 4n-1$, let 
      \begin{equation} \label{v-1}
        \begin{aligned}
      v_l :=&2\sum_{k=1}^{4n-l}(-1)^{T_{2n}(k)+kl+k+l+1}E_{k'\,4n+1-l-k}. \\
        \end{aligned}
      \end{equation}
      Then $v_l=0 \text{ if }l \equiv 1 \text{ or }2 \Mod 4$ and $v_l \neq0 \text{ if }l \equiv 0 \text{ or }3 \Mod 4$.
      \item For an integer $l$ such that $1 < l \leq 2n$, let $t=2n-l-2$ and   
      \begin{equation} \label{v-2}
        \begin{aligned}
        v_l :=& (-1)^t\left(E_{2n\,t+1}-E_{2n+1\,t+1}\right)+(-1)^{T_{2n}(t+1)+1}\left(E_{(t+1)'\,(2n)'}-E_{(t+1)'\,2n}\right)\\
        &+\sum_{k=1}^{t}(-1)^{t(k+1)}E_{2n-1-t+k\,k}+\sum_{k=1}^{t}(-1)^{tk+1+T_{2n-1-t+k}(k)}E_{k'\,(2n-1-t+k)'}\\
        &+2\sum_{k=t+1}^{2n}(-1)^{T_{2n}(k)+kl+k+l}E_{k'\,l+2-k}
        \end{aligned}
      \end{equation}
      Then $v_l=0 \text{ if }l \equiv 1 \text{ or }2 \Mod 4$ and $v_l \neq0 \text{ if }l \equiv 0 \text{ or }3 \Mod 4$.
      \item Let $\tilde{v}_{2n}:=E_{2n\, 1}-E_{2n+1\,1}$. Then for $v_l$'s in \eqref{v-1} and \eqref{v-2}, 
      \begin{equation}
        V^f:=\{v_l\,|\, 1 < l \leq 4n-1,\ l\equiv 0 \text{ or }3 \Mod 4\} \cup \{\widetilde{v}_{2n}\}
      \end{equation}
      forms a basis of $\g^{f}$.
    \end{enumerate}
  \end{lemma}
  \begin{proof}
    (1) and (2) follow from direct computations. Let us show (3). Since $v_l$'s for $1 < l \leq 4n-1$ are in distinct homogeneous spaces with respect to the $\frac{\Z}{2}$-grading on $\g$, the linear independency of $V^f$ is clear. Also, $|V^f|=2n=\dim{\g^f}$, so we only need to show that each $v_l\in V^f$ and $\widetilde{v}_{2n}$ are indeed in $\g^f$.   Again, by direct computations, we get $[f, v_l]=[f, \widetilde{v}_{2n}]=0$.
  \end{proof}
  
  Now we need to check if the elements of $V^f$ satisfy the conditions of Proposition \ref{Prop:property_generator}. In particular, to check the condition (i), define a gradation $\Delta$ on $\mathcal{V}(\bar{\mathfrak{p}})(\!(D^{-1})\!)$ by 
\begin{equation} \label{eqn:conf wt extension}
      \Delta_{\bar{a}} = \frac{1}{2}-j_a, \quad \Delta_D= \frac{1}{2}, \quad \Delta_{D^{-1}}= -\frac{1}{2},\quad \Delta_{AB}= \Delta_A + \Delta_B
\end{equation}
where $a \in \g_{j_a}$ and $A,B\in \mathcal{V}(\bar{\mathfrak{p}})$. Note that on $\mathcal{V}(\bar{\frak{p}})$, $\Delta$ coincides with the gradation defined in \eqref{Eqn:conformal weight}.

\begin{lemma}\label{weight}
Recall the elements $w_1, \cdots, w_{4n-1}, \widetilde{w}_{2n}$ of $\mathcal{W}(\bar{\g},f)$ in Theorem \ref{thm:osp 4n elements} and the elements $v_l$'s and $\widetilde{v}_{2n}$ of $V^f$ in Lemma \ref{Lem:osp basis of keradf}. Then we have 
\begin{enumerate}
 \item $\Delta_{\bar{v}_{l}}=\frac{l}{2}$ and $\Delta_{\overline{\widetilde{v}}_{2n}}=n$,
 \item $\Delta_{w_{l}}= \frac{l}{2}$ for $l=1.\cdots,4n-1$ and $\Delta_{\widetilde{w}_{2n}}=n$.
 \end{enumerate}
    \end{lemma}
    \begin{proof}
    For $i\in I$, let 
    \begin{equation}
  \bar{\imath}= \left\{\begin{array}{ll} i & \text{ if } 1\leq{i}\leq{2n}, \\ i-1 & \text{ if } 2n+1\leq{i}\leq{4n}.\end{array} \right.
\end{equation}
Then the degree of  
$F_{ij}$ is $\frac{\bar{\jmath}-\bar{\imath}}{2}$ with respect to the $\frac{\Z}{2}$-grading of $\g$, so $\Delta_{F_{ij}}=\frac{1-\bar{\jmath}+\bar{\imath}}{2}$. Hence, (1) follows. Now recall that $rdet\big(\mathcal{A}_{\frak{osp}}^{(2n|2n)}\big)$ can be computed as \eqref{cdet osp(2n|2n)}. Using Lemma \ref{Lemma A_k, B_k}, one can inductively show that $\Delta_{A_k}=\frac{k}{2}=\Delta_{B_k}$ for $k=1,\cdots,2n$. Therefore, $\Delta_{rdet\left(\mathcal{A}_{\frak{osp}}^{(2n|2n)}\right)}=\frac{4n-1}{2}$ and so $\Delta_{w_l}=\frac{l}{2}$. Since $\widetilde{w}_{2n}=a_{2n}$, we also have $\Delta_{\widetilde{w}_{2n}}=n$.
    \end{proof}

  \begin{theorem} \label{Thm:osp(2n|2n) generators}
    As in \eqref{2n|2n, row determinant-2}, write
    \[rdet(\mathcal{A}_{\mathfrak{osp}}^{(2n|2n)})=\sum_{l=0}^{4n-1}w_{4n-1-l} D^l +(-1)^{n}\widetilde{w}_{2n}D^{-1}\widetilde{w}_{2n}.\]
    Then the set
    $\{w_l\,|\, 1\leq l \leq 4n-1, i\equiv 0 \text{ or }3\Mod 4\} \cup \{\widetilde{w}_{2n}\}$
    freely generates $\mathcal{W}(\bar{\g}, f)$.
  \end{theorem}
  \begin{proof}
  We use Proposition \ref{Prop:property_generator}. By Lemma \ref{Lem:osp basis of keradf} and Lemma \ref{weight}, 
it is enough to show that 
\[ w_l^1=\bar{v}_l, \quad w_{2n}^1= \bar{\tilde{v}}_{2n}\]
where $w_l^1\in \bar{\g}$ and $w_{2n}^1$ are  the linear parts  of $w_l$ and $w_{2n}$ which are not  total derivatives.
 
For $0 \leq l < 2n-1$, the element  $w_{4n-1-l}^1\in \bar{\g}$ is in the coefficient of $D^l$ of the following sum:
\begin{equation} \label{coeff D^l}
  \  \    2\sum_{j+k=l+2}(-1)^{T_{2n}(k)+k}D^{j-1}\bar{E}_{k'\,j}D^{k-1},
\end{equation}
    where the sum is taken over $j,k\in\{1,2, \cdots,2n\}$.
    A simple computation shows that the linear part in \eqref{coeff D^l} is 
$2\sum_{j+k=l+2}(-1)^{T_{2n}(k)+kj+j}\bar{E}_{k'\,j}+\left(\text{a total derivative part}\right).$
Therefore, $w_{4n-1-l}^1=  2\sum_{j+k=l+2}(-1)^{T_{2n}(k)+kj+j}\bar{E}_{k'\,j}=\bar{v}_{4n-1-l}$.

    Similarly, in case of $2n-1 \leq l < 4n-2$, consider the coefficient of $D^l$ of
    \[   \  \  A_{2n}D^{-1}B_{2n}+2\sum_{j+k=l+2}(-1)^{T_{2n}(k)+k}D^{j-1}\bar{E}_{k'\,j}D^{k-1}\]
    and we get $w_{4n-1-l}^1=\bar{v}_{4n-1-l}$.

    Finally, it is clear that $\overline{\tilde{v}}_{2n}=\widetilde{w}^1_{2n}$. Hence we are done.
  \end{proof}

\subsection{ $\mathcal{W}(\overline{\mathfrak{osp}}(2n+2|2n),f)$} \label{Sec:W(osp(2n+2|2n)}\

Let $\g= \mathfrak{osp}(2n+2|2n)$ and take
\begin{equation}
E_{ij}:=\frac{1}{2}(-1)^{p(i)}F_{ij} \text{ for } i,j\in I.
\end{equation}
 To make an analogous matrix to \eqref{matix(2n|2n)}, we consider the $(4n+3)\times(4n+3)$ matrix $ \mathcal{A}_{\mathfrak{osp}}^{(2n+2|2n)}$ whose entries are in $\mathcal{V}(\bar{\mathfrak{p}})(\!(D^{-1})\!)$. Precisely,  
  \begin{equation} \label{matix(2n+2|2n)}
 \mathcal{A}_{\mathfrak{osp}}^{(2n+2|2n)}=\small{
\begin{pmatrix}
\begin{matrix}
&&\\
& \mu_{11}&\\
&&\\
\end{matrix}
  & \rvline & 0& \rvline&\begin{matrix}
&&\\
& \mu_{12}&\\
&&\\
\end{matrix} \\
\hline
0 & \rvline &
D^{-1}&
  \rvline & 0\\
\hline
\begin{matrix}
&&\\
& \mu_{21}&\\
&&\\
\end{matrix}
& \rvline & 0
&
  \rvline & 
\begin{matrix}
&&\\
& \mu_{22}&\\
&&\\
\end{matrix}
\end{pmatrix}
}
\end{equation}
where the entries of the $(2n+1)\times (2n+1)$ matrices $\mu_{11}$, $\mu_{22}$, $\mu_{12}$ and $\mu_{21}$ are given by 
\begin{eqnarray}
&&\quad(\mu_{11})_{ij}= \left\{\begin{array}{ll} \delta_{2n\, i} D + \bar{E}_{2n+1\, i}-\bar{E}_{2n+2\, i} & \text{ if } i\leq j=2n+1, \\ \delta_{ij}D+ \bar{E}_{ji} & \text{ if } i\leq j <2n+1, \\ -\delta_{i-1\, j}; & \text{ otherwise},
\end{array} \right.
\\
&&\quad(\mu_{22})_{2n+2-i\ 2n+2-j}= \left\{\begin{array}{ll} \delta_{2n+2\, j} D + \bar{E}_{j'\, (2n+1)'}-\bar{E}_{j'\, (2n+2)'} & \text{ if } j\leq i=2n+1, \\ \delta_{ij}D+ \bar{E}_{j'i'} & \text{ if } j\leq i <2n+1, \\ (-1)^i\delta_{i\, j-1} & \text{ otherwise},
\end{array} \right.
\\
&&\quad(\mu_{12})_{i\ 2n+2-j}= \bar{E}_{j'\, i}\quad \text{and} \quad (\mu_{21})_{i\, j}= -2\, \delta_{i\, 1} \delta_{j\, 2n+1} D
\end{eqnarray}
with $i'= 4n+3-i\in I$.

Denote the row determinant of the upper-left $k \times k$ submatrix of $\mu_{11}$ by $A_k$, and that of the lower-right $k \times k$ submatrix of $\mu_{22}$ by $B_k$. Let $A_0=B_0=1$. Then the row determinant of $\mathcal{A}_{\mathfrak{osp}}^{(2n+2|2n)}$ is 
\begin{equation} \label{cdet osp(2n+2|2n)}
  \begin{aligned}
  rdet(\mathcal{A}_{\mathfrak{osp}}^{(2n+2|2n)})&=A_{2n+1}D^{-1}B_{2n+1}+2\sum_{j,k=1}^{2n+1}{(-1)^{n+\floor*{\frac{k+1}{2}}}A_{j-1}\bar{E}_{k'j}B_{k-1}}.
  \end{aligned}
\end{equation}
If we denote 
\[ A_{2n+1}= \sum_{i=0}^{2n+1} a_{2n+1-i} D^i\]
for $a_{2n+1-i} \in \mathcal{V}(\bar{\mathfrak{p}})$, then the row determinant of $\mathcal{A}_{\mathfrak{osp}}^{(2n+2|2n)}$ is
\begin{equation}\label{2n+2|2n, row determinant-2}
  rdet(\mathcal{A}_{\mathfrak{osp}}^{(2n+2|2n)})=D^{4n+1}+ \sum_{k=0}^{4n}w_{4n+1-k} D^k +(-1)^{n+1}a_{2n+1}D^{-1}a_{2n+1}
\end{equation}
for $w_1, w_2, \cdots, w_{4n+1}\in\mathcal{V}(\bar{\mathfrak{p}})$.

Using an analogous argument to Section \ref{Sec:W(osp(2n|2n))}, one can show the following theorem.

  \begin{theorem}
    Let $w_1,\cdots, w_{4n+1}$ be as in \eqref{2n+2|2n, row determinant-2} and let $\widetilde{w}_{2n+1}:=a_{2n+1}$. Then 
    \begin{enumerate}
    \item $w_1, \cdots, w_{4n+1},\widetilde{w}_{2n+1}\in\mathcal{W}(\bar{\g},f)$,
    \item the set
    $\{w_l\,|\, 1\leq l \leq 4n+1, i\equiv 0 \text{ or }3\Mod 4\} \cup \{\widetilde{w}_{2n+1}\}$
    freely generates $\mathcal{W}(\bar{\g}, f)$.
    \end{enumerate}
  \end{theorem}
\begin{proof}

 One can show that 
 \[\text{ad}_\chi^{\mathfrak{J}}\bar{F}_{i\,i+1}( w_k)=0 \quad \text{ and }\quad \text{ad}_\chi^{\mathfrak{J}}\bar{F}_{2n+2\,2n}( w_k )=0\]
 for $i=1,\cdots,2n$ as in Theorem \ref{thm:osp 4n elements}. Since the set $\{F_{i\,i+1}\,|\,i=1,2,\cdots,{2n}\}\cup \{F_{2n+2\,2n}\}$ generates $\n$, if follows that $w_1,\cdots,w_{4n+1}\in \mathcal{W}(\bar{\g},f)$. Also, one can show that $\widetilde{w}_{2n+1}\in \mathcal{W}(\bar{\g},f)$ as in Lemma \ref{Lem:bracket for submatrix}.
 \\We use Proposition \ref{Prop:property_generator} once more to prove (2).
    We have that for $-n\leq k\leq 0$,
    \begin{equation}
      \dim\g_k=\left\{\begin{array}{lll} k+2n+\frac{3}{2} & \text{ if } 2k\text{ is odd, } \\
          k+2n+1 & \text{ if } 2k\equiv 0 \Mod4, \\
          k+2n+2 & \text{ if } 2k\equiv 2 \Mod4, \end{array}\right.
    \end{equation}
  and for $-2n\leq k <-n$,
  \begin{equation}
      \dim\g_k=\left\{\begin{array}{lll} k+2n+\frac{1}{2} & \text{ if } 2k\text{ is odd, } \\
          k+2n & \text{ if } 2k\equiv 0 \Mod4, \\
          k+2n+1 & \text{ if } 2k\equiv 2 \Mod4. \end{array}\right.
    \end{equation}
Hence, one can find a basis
    \begin{equation}
        V^f=\big\{v_l\in\g_{\frac{1-l}{2}}\,|\, 1 < l \leq 4n+1,\ l\equiv 0\text{ or }3 \Mod 4\big\} \cup \big\{\tilde{v}_{2n+1}\in \g_{-n}\big\}
      \end{equation}
  of $\g^f$, whose elements are all homogeneous with respect to the $\frac{\Z}{2}$-grading on $\g$.
  To check the conditions of Proposition \ref{Prop:property_generator}, define a gradation $\Delta$ on $\mathcal{V}(\bar{\mathfrak{p}})(\!(D^{-1})\!)$ as in \eqref{eqn:conf wt extension}.
 If we denote
    \begin{equation}
  \bar{\imath}= \left\{\begin{array}{ll} i+1 & \text{ if } 1\leq{i}\leq{2n+1}, \\ i & \text{ if } 2n+2\leq{i}\leq{4n+2},\end{array} \right.
\end{equation}
for $i\in I$ then $\Delta_{F_{ij}}=\frac{1-\bar{\jmath}+\bar{\imath}}{2}$ for $i,j\in I$. Therefore, $\Delta_{rdet(\mathcal{A}_{\mathfrak{osp}}^{(2n+2|2n)})} = \frac{4n+1}{2}$ and $\Delta_{A_{2n}}=\frac{2n+1}{2}$. In consequence, $\Delta_{w_l}= \frac{l}{2}$ for $l=1,\cdots, 4n+1$ and $\Delta_{\widetilde{w}_{2n+1}}=\frac{2n+1}{2}$. Now, one can properly choose elements of $V^f$ so that 
  \[ w_l^1=\bar{v}_l,\quad \widetilde{w}^1_{2n+1}=\bar{\tilde{v}}_{2n+1},\]
  where $w_l^1,  \widetilde{w}^1_{2n+1}\in \bar{\g}$ are  the linear parts  of $w_l$,   $\widetilde{w}_{2n+1}$ which are not  total derivatives. Since $\Delta_{\bar{v}_l}=\frac{l}{2}$ and $\Delta_{\bar{\tilde{v}}_{2n+1}}=\frac{2n+1}{2}$, the assertion (2) follows.
\end{proof}

\end{document}